\documentclass[11pt, letterpaper]{article}
\usepackage[margin=1in]{geometry}

\newcommand*\samethanks[1][\value{footnote}]{\footnotemark[#1]}

\title{A PTAS for the Steiner Forest Problem in Doubling Metrics}
\date{}

\author{T-H. Hubert Chan\thanks{Department of Computer Science, the University of Hong Kong. {\texttt{\{hubert,sghu,sfjiang\}@cs.hku.hk}}} \and Shuguang Hu\samethanks \and  Shaofeng H.-C. Jiang\samethanks}

\newif\ifplain

\plaintrue

\usepackage{amsmath}
\usepackage{amssymb}
\usepackage{amsthm}

\usepackage{url}
\usepackage{color}
\usepackage{paralist}
\usepackage{hyperref}
\usepackage{xspace}

\newtheorem{fact}{Fact}[section]
\newtheorem{claim}{Claim}[section]

\ifplain
\newtheorem{theorem}{Theorem}[section]
\newtheorem{lemma}{Lemma}[section]
\newtheorem{corollary}{Corollary}[section]
\newtheorem{definition}{Definition}[section]
\fi

\newcommand{\eps}{\epsilon}

\newcommand{\enr}{\eps}

\newcommand{\ep}{\eps}

\newcommand{\ceil}[1]{\ensuremath{\lceil #1 \rceil}}

\newcommand{\poly}{\operatorname{poly}}

\newcommand{\TSP}{\ensuremath{\mathsf{TSP}}\xspace}

\newcommand{\SFP}{\ensuremath{\mathsf{SFP}}\xspace}
\newcommand{\STP}{\ensuremath{\mathsf{STP}}\xspace}
\newcommand{\INS}{\ensuremath{W}\xspace}

\newcommand{\OPT}{\ensuremath{\mathsf{OPT}}\xspace}

\newcommand{\T}{\ensuremath{\mathsf{T}}\xspace}
\newcommand{\heur}{\ensuremath{\mathsf{H}}\xspace}

\newcommand{\Exp}{\ensuremath{\mathsf{Exp}}\xspace}
\newcommand{\ALG}{\ensuremath{\mathsf{ALG}}\xspace}
\newcommand{\DP}{\ensuremath{\mathsf{DP}}\xspace}
\newcommand{\Diam}{\ensuremath{\mathsf{Diam}}}
\newcommand{\expct}[1]{\ensuremath{\text{{\bf E}$\left[#1\right]$}}}

\newcommand{\ignore}[1]{}

\newcommand{\Ht}{\ensuremath{\mathsf{ht}}\xspace}
\newcommand{\Par}{\ensuremath{\mathsf{par}}\xspace}
\newcommand{\Des}{\ensuremath{\mathsf{des}}\xspace}

\newcommand{\Bas}{\ensuremath{\mathsf{Bas}}\xspace}
\newcommand{\Pro}{\ensuremath{\mathsf{Pro}}\xspace}
\newcommand{\Vir}{\ensuremath{\mathsf{Vir}}\xspace}
\newcommand{\NBas}{\ensuremath{\mathsf{NBas}}\xspace}
\newcommand{\Can}{\ensuremath{\mathsf{Can}}\xspace}
\newcommand{\Eff}{\ensuremath{\mathsf{Eff}}\xspace}
\newcommand{\Dis}{\ensuremath{\mathsf{Dis}}\xspace}

\newcommand{\BAS}{\ensuremath{\mathsf{BAS}}\xspace}
\newcommand{\NBAS}{\ensuremath{\mathsf{NBAS}}\xspace}
\newcommand{\Val}{\ensuremath{\mathsf{val}}\xspace}

\begin{document}

\begin{titlepage}

\maketitle

\begin{abstract}

We achieve a (randomized) polynomial-time approximation	scheme (PTAS) for the Steiner Forest Problem in doubling metrics. Before our work, a PTAS is given only for the Euclidean plane in [FOCS 2008: Borradaile, Klein and Mathieu].  Our PTAS also shares similarities with the dynamic programming for sparse instances used in [STOC 2012: Bartal, Gottlieb and Krauthgamer] and [SODA 2016: Chan and Jiang]. However, extending previous approaches requires overcoming several non-trivial hurdles, and we make the following technical contributions.

(1) We prove a technical lemma showing that Steiner points have to be ``near'' the terminals in an optimal Steiner tree. This enables us to define a heuristic to estimate the local behavior of the optimal solution, even though the Steiner points are unknown in advance.	This lemma also generalizes previous results in the Euclidean plane, and may be of independent interest for related problems involving Steiner points.

(2) We develop a novel algorithmic technique known as ``adaptive cells'' to overcome the difficulty of keeping track of multiple components in a solution. Our idea is based on but significantly different from the previously proposed ``uniform cells'' in the FOCS 2008 paper, whose techniques cannot be readily applied to doubling metrics.

\end{abstract}

\thispagestyle{empty}
\end{titlepage}

\section{Introduction}
\label{sec:intro}
We consider the Steiner Forest Problem ($\SFP$) in a metric space~$(X, d)$.
An instance of the problem is given by 
a collection~$\INS$ of~$n$ terminal pairs~$\{(a_i, b_i) : i\in [n]\}$ in~$X$,
and 
the objective is to find a minimum weight graph~$F = (V,E)$
(where~$V$ is a subset of~$X$ and the edge weights are induced 
by the metric space)
such that every pair in~$\INS$ is connected in~$F$.

\subsection{Problem Background}

\ignore{
The general version of~$\SFP$ is defined on a undirected graph.
However, since there exists a polynomial time ratio preserving reduction,
it is sufficient to consider the metric~$\SFP$ problems.}

The problem is well-known in the computer science community.
In general metrics, 
Chleb{\'{\i}}k and Chleb{\'{\i}}kov{\'{a}}~\cite{DBLP:journals/tcs/ChlebikC08}
showed that~$\SFP$ is \textbf{NP}-hard to approximate with ratio better than~$\frac{96}{95}$.
The best known approximation ratio achievable in polynomial time is~$2$~\cite{DBLP:journals/siamcomp/GoemansW95,DBLP:journals/siamcomp/AgrawalKR95}.
Recently, Gupta and Kumar~\cite{DBLP:conf/stoc/Gupta015} gave a purely combinatorial
greedy-based algorithm that also achieves constant ratio.
However, it is still an open problem to break the~$2$-approximation barrier in general metrics for $\SFP$.

\noindent \textbf{$\SFP$ in Euclidean Plane and Planar Graphs.}
In light of the aforementioned hardness result~\cite{DBLP:journals/tcs/ChlebikC08},
restrictions are placed on the metric space to achieve
$(1+\epsilon)$ approximation in polynomial time.
In the Euclidean plane,
a randomized polynomial-time approximation scheme (PTAS) was obtained in \cite{DBLP:conf/focs/BorradaileKM08},
using the dynamic programming framework proposed by Arora~\cite{DBLP:journals/jacm/Arora98}.
Later on, a simpler analysis was presented in \cite{DBLP:journals/algorithmica/BateniH12},
in which a new structural property is proved and additional information is incorporated in the dynamic programming algorithm. It was only suggested
that similar techniques might be applicable to higher-dimensional Euclidean space.

Going beyond the Euclidean plane,
a PTAS for planar graphs is obtained in \cite{DBLP:journals/jacm/BateniHM11} and more generally, on bounded genus graphs.
As a building block, they also obtained
a PTAS for graphs with bounded treewidth.

\noindent\textbf{Steiner Tree Problems.}
A notable special case of~$\SFP$ is the Steiner Tree Problem ($\STP$), in which all terminals
are required to be connected. In general metrics, the MST on the terminal points
simply gives a~$2$-approximation. There is a long line of research to improve the~$2$-approximation, 
and the state-of-the-art approximation ratio
$1.39$ was presented in \cite{DBLP:conf/stoc/ByrkaGRS10} via an LP rounding approach.
On the other hand, it is \textbf{NP}-hard to approximate \STP better than the ratio~$\frac{96}{95}$~\cite{DBLP:journals/tcs/ChlebikC08}.

For the group $\STP$ in general metrics, it is \textbf{NP}-hard
to approximate within~$\log^{2-\epsilon}{n}$ \cite{DBLP:conf/stoc/HalperinK03} unless~$\textsf{NP} \subseteq \textsf{ZTIME}(n^{\text{polylog}(n)})$.
On the other hand, it is possible to approximate within~$O(\log^3{n})$ as shown in \cite{DBLP:journals/jal/GargKR00}. Restricting to planar graphs,
the group $\STP$
can be approximated within $O(\log {n} \, \text{poly}\log{\log{n}})$~\cite{DBLP:journals/talg/DemaineHK14}, and very recently,
this result is improved to a PTAS~\cite{DBLP:conf/stoc/BateniDHM16}.

For more related works, we refer the reader to a survey
by Hauptmann and Karpi{\'n}ski~\cite{hauptmann2013compendium}, who gave a comprehensive literature
review of $\STP$ and its variations.

\noindent \textbf{PTAS's for Other Problems in Doubling Metrics.}
Doubling dimension captures the local growth rate of a
metric space. A~$k$-dimensional Euclidean dimension has doubling dimension~$O(k)$.
A challenge in extending algorithms for low-dimensional Euclidean space
to doubling metrics is the lack of geometric properties in doubling metrics.
Although QPTAS's for various approximation problems in doubling metrics, such as
the Traveling Salesman Problem ($\TSP$)
and $\STP$, were presented
in \cite{DBLP:conf/stoc/Talwar04}, a PTAS was only recently achieved for $\TSP$~\cite{DBLP:conf/stoc/BartalGK12}. Subsequently,
a PTAS is also achieved for group $\TSP$
in doubling metrics~\cite{DBLP:conf/soda/ChanJ16}. 
Before this work, the existence of a PTAS for $\SFP$ (or even the special case $\STP$) in doubling metrics remains an open problem.

\subsection{Our Contribution and Techniques}

Although PTAS's for $\TSP$ (and its group variant) are known, as we shall explain later, the nature of $\SFP$ and 
$\TSP$-related problems are quite different.
Hence, it is interesting to investigate what new techniques are
required for $\SFP$.
Fundamentally, it is an important question that whether the notion of doubling dimension
captures sufficient properties of a metric space to
design a PTAS for $\SFP$, even without the geometric properties
that are crucially used in obtaining approximation schemes for $\SFP$ in
the Euclidean plane~\cite{DBLP:conf/focs/BorradaileKM08}.

In this paper, we settle this open problem by
giving a (randomized) PTAS for
$\SFP$ in doubling metrics.
We remark that previously even a PTAS for $\SFP$ in higher-dimensional Euclidean space is 
not totally certain.

\begin{theorem}[PTAS for~$\SFP$ in Doubling Metrics]
	\label{theorem:main}
	For any~$0<\epsilon<1$, there is a (randomized) algorithm that
	takes an instance of~$\SFP$ with~$n$ terminal pairs
	in a metric space with doubling dimension at most~$k$,
	and 
	returns
	a~$(1+\epsilon)$-approximate solution with
	constant probability, running in time
	$O(n^{O(1)^k}) \cdot \exp(\sqrt{\log{n}}\cdot O(\frac{k}{\epsilon})^{O(k)})$.
\end{theorem}

We next give an overview of our techniques.
On a high level, we use the divide and conquer
framework that was originally used
by Arora~\cite{DBLP:journals/jacm/Arora98}
to achieve a PTAS for $\TSP$ in Euclidean space,
and was extended recently to doubling metrics~\cite{DBLP:conf/stoc/BartalGK12}. 

However, we shall explain that it is non-trivial to adapt
this framework to $\SFP$, and how we overcome the difficulties encountered. 
Moreover, we shall provide
some insights regarding the relationship between Euclidean and doubling metrics,
and discuss the implications of our technical lemmas.

\noindent\textbf{Summary of Framework.}
As in~\cite{DBLP:conf/stoc/BartalGK12},
a PTAS is designed for a class of special instances known as
\emph{sparse} instances. Then, it can be shown that the general instances can be
decomposed into sparse instances. Roughly speaking, an instance is sparse,
if there is an optimal solution such that for any ball~$B$ with radius~$r$,
the portion of the solution in~$B$ has weight that is small with respect to~$r$.

The PTAS for the sparse instances is usually based on a dynamic program,
which is based on a randomized hierarchical decomposition
as in~\cite{DBLP:conf/stoc/Talwar04,DBLP:conf/stoc/BartalGK12}.
This framework has also been successfully applied
to achieve a PTAS for group $\TSP$ in doubling metrics~\cite{DBLP:conf/soda/ChanJ16}.
Intuitively, sparsity is
used to establish the property that with high enough probability,
a cluster in the randomized decomposition cuts a (near) optimal tour only a small number of times~\cite[Lemma 3.1]{DBLP:conf/stoc/BartalGK12}.
However, $\SFP$ brings new significant challenges when such a framework is applied.  We next
describe the difficulties and give an overview of our technical contributions.

\vspace{10pt}
 
\noindent \textbf{Challenge 1: It is difficult to
detect a sparse instance because which Steiner points are used
by the optimal solution are unknown.}
Let us first consider $\STP$, which is a special case of $\SFP$ in which
all (pairs of) terminals are required to be connected. In other words,
the optimal Steiner tree is the minimum weight graph that connects all terminals.
Unlike $\TSP$ in which the points visited by a tour are clearly known in advance,
it is not known which points will be included in the optimal Steiner tree.



In~\cite{DBLP:conf/stoc/BartalGK12},
a crucial step is to estimate the sparsity of a ball~$B$,
which measures the weight of the portion of the optimal solution
restricted to~$B$.
For $\TSP$ tour, this can be estimated from the points inside~$B$ that have
to be visited.
However, for solution involving Steiner points,
it is difficult to analyze the solution inside some ball~$B$,
because it is possible that there are few (or even no) terminals inside~$B$,
but the optimal solution could potentially have lots of Steiner points
and a large weight inside~$B$.


\noindent \textbf{Our Solution: Analyzing the Distribution of Steiner Points in an Optimal Steiner Tree in Doubling Metrics.}
We resolve this issue by showing a technical characterization
of Steiner points in an optimal Steiner tree for doubling
metrics. This technical lemma is used crucially in our proofs,
and we remark that it could be of interest for other problems involving Steiner points in doubling metrics.


\begin{lemma}[Formal version in Lemma~\ref{lemma:near_terminal}]
\label{lemma:informal_near_terminal}
For a terminal set~$S$ with diameter~$D$, if an optimal Steiner tree spanning~$S$ has no edge longer than~$\gamma D$,
then every Steiner point in the solution is within~$O(\sqrt{\gamma}) \cdot  D$ distance to some terminal in~$S$,
where the big~$O$ hides the
dependence on the doubling dimension.
\end{lemma}

We observe that variants of Lemma~\ref{lemma:informal_near_terminal}
have been considered on the Euclidean plane.
In \cite{du1985steiner, du1987steiner},
it is shown that if the terminal set consists of~$n$ evenly distributed points on a unit circle,
then for large enough~$n$, there is no Steiner points in an optimal Steiner tree.
To see how this relates to our lemma,
when~$n$ is sufficiently large,
it follows that adjacent points in the circle are
very close to each other.  Hence,
any long edge in a Steiner tree could be replaced by some short
edge between adjacent terminals in the circle.
Our lemma then implies that all Steiner points must be near
the terminals, which is a weaker statement
than the conclusion in~\cite{du1985steiner}, but
is enough for our purposes.
We emphasize that the results in \cite{du1985steiner,du1987steiner} rely on the geometric properties
of the Euclidean plane.
However, in our lemma, we only use that the doubling dimension is bounded.


\noindent \textbf{\emph{Implication of Lemma~\ref{lemma:informal_near_terminal}
on Sparsity Heuristic.}}
We next demonstrate an example of how we
use this technical lemma.
In Lemma~\ref{lemma:heuristic_small},
we argue that
our sparsity heuristic provides an upper bound
on the weight of the portion of an optimal solution~$F$
within some ball~$B$.

The idea is that we remove the edges in~$F$ within~$B$
and add back some edges of small total weight to maintain connectivity.
We first add a minimum spanning tree~$H$ on
some net-points~$N$ within~$B$ of an appropriate scale~$\gamma \cdot D$.
Using the property of doubling dimension, we argue that the number
of points in~$H$ is bounded and so is its weight.
In one of our case analysis, there are two sets~$S$ and~$T$ of terminals
that are far apart~$d(S,T) \geq D$, and we wish to argue that
in the optimal Steiner tree~$F$ connecting~$S$ and~$T$,
there is an edge~$\{u,v\}$ of length at least~$\Omega(\gamma) \cdot D$.
If this is the case, we could remove this edge and connect~$u$ and~$v$
to their corresponding net-points directly.
For contradiction's sake, we assume there is no such edge, but 
Lemma~\ref{lemma:informal_near_terminal} implies that every Steiner point
must be close to either~$S$ and~$T$.  Since~$S$ and~$T$ are far apart,
this means that there is a long edge after all.

Conversely, in Lemma~\ref{lemma:heuristic_upper_bound},
we also use this technical lemma to show that
if the sparsity heuristic for some ball~$B$ is large,
then the portion of the optimal solution~$F$ inside~$B$ is also large.

%

\vspace{15pt}

\noindent \textbf{Challenge 2: In doubling metrics,
the number of cells for keeping track of connectivity
in each cluster could be too large.}
Unlike the case for $\TSP$ variants~\cite{DBLP:conf/stoc/BartalGK12, DBLP:conf/soda/ChanJ16},
the solution for~$\SFP$ need not be connected.
Hence, in the dynamic programming algorithm for~$\SFP$,
in addition to keeping track of what \emph{portals} are used
to connect a cluster to points outside,
we need to keep information on which portals
the terminals inside a cluster are connected to.  
In previous works~\cite{DBLP:conf/focs/BorradaileKM08},
the notion of \emph{cells} is used for this purpose.


\noindent \emph{Previous Technique: Cell Property.}
The idea of \emph{cell property} was first introduced in~\cite{DBLP:conf/focs/BorradaileKM08},
which gave a PTAS for $\SFP$ in the Euclidean plane using dynamic programming.
Since there would have been an exponential number of dynamic program entries
if we keep information on which portal is used by every terminal to connect 
to its partner outside the cluster,
the high level idea is to partition a cluster
into smaller clusters (already provided by the hierarchical decomposition) known as \emph{cells}.
Loosely speaking, the \emph{cell property} 
ensures that every terminal inside the same cell
must be connected to points outside the cluster
in the same way.  More precisely,
a solution~$F$ satisfies the cell property
if for every cluster~$C$ and every cell~$e$ inside~$C$,
there is only one component in the portion of~$F$ restricted to~$C$
that connects~$e$ to points outside~$C$.

A great amount of work was actually
needed in~\cite{DBLP:conf/focs/BorradaileKM08}
and subsequent work~\cite{DBLP:journals/algorithmica/BateniH12}
to show that it is enough to consider cells whose diameters
are constant times smaller than that of its cluster.
This allows the number of dynamic program entries to be bounded,
which is necessary for a PTAS.

\noindent \textbf{Difficulty Encountered for Doubling Metrics.}
When the notion of cell is applied to the dynamic program
for $\SFP$ in doubling metrics, an important issue is that
the diameters of cells need to be about~$\Theta(\log n)$ times
smaller than that of its cluster,  because there are around~$\Theta(\log n)$
levels in the hierarchical decomposition.
Hence, the number of cells in a cluster is~$\Omega(\text{poly}\log{n})$,
which would eventually 
lead to a QPTAS only.  A similar situation is observed 
when dynamic programming was first used for $\TSP$ on doubling metrics~\cite{DBLP:conf/stoc/Talwar04}.  However,
the idea of using sparsity as in~\cite{DBLP:conf/stoc/BartalGK12}
does not seem to immediately provide a solution.


\noindent\textbf{Our Solution: Adaptive Cells.}
Since there are around~$\Theta(\log n)$ levels
in the hierarchical decomposition,
it seems very difficult to increase the diameter of cells in a cluster.
Our key observation is that the cells are needed only for covering
the portion of a solution inside a cluster that touches the cluster boundary.
Hence, we use the idea of \emph{adaptive cells}.  Specifically,
for each connected component~$A$ in the solution crossing a cluster~$C$,
we define the corresponding \emph{basic cells} such that if the component~$A$
has larger weight, then its corresponding basic cells (with respect to cluster
$C$) will have larger diameters.
Combining with the notion of sparsity and bounded doubling dimension,
we can show that we only need to pay attention to a small number of cells.


%

\noindent \textbf{Further Cells for Refinement.}
Since the dynamic program entries are defined
in terms of the hierarchical decomposition
and the entries for a cluster are filled recursively with respect
to those of its child clusters,
we would like the cells to have a \emph{refinement} property, i.e.,
if a cluster~$C$ has some cell~$e$ (which itself is some descendant cluster of~$C$),
then the child~$C'$ containing~$e$ has either~$e$ 
or all children of~$e$ as its cells.

At first glance, a quick fix may be to push down each basic cell in~$C$
to its child clusters. Although we could still bound the number of relevant cells,
it would be difficult to bound the cost to achieve the cell property. The reason is that 
the basic cells from higher levels are too large for the descendant clusters.
When more than one relevant component intersects such a large cell,
we need to add edges to connect the components.
However, if the diameter of the cell is too large compared to the cluster,
these extra edges would be too costly.


We resolve this issue by introducing \emph{non-basic} cells for a cluster:
\emph{promoted} cells and \emph{virtual} cells.  These cells are introduced to 
ensure that every sibling of a basic cell is present.  Moreover,
only non-basic cells of a cluster will be passed to its children.
We show in Lemma~\ref{lemma:bound_eff} that
the total number of \emph{effective} cells for a cluster
is not too large. Moreover, Lemma~\ref{lemma:refinement}
shows that the refinement property still holds even if we only pass
the non-basic cells down to the child clusters.
More importantly, we show that as long as we enforce the
cell property for the basic cells, 
the cell property for all cells are automatically ensured.
This means that it is sufficient to bound the cost to achieve the
cell property with respect 
to only the basic cells.


\noindent\textbf{Further Techniques: Global Cell Property.}
We note that the cell property in~\cite{DBLP:conf/focs/BorradaileKM08} is localized.
In particular, for each cluster~$C$, we restrict the solution inside~$C$, which
could have components disconnected within~$C$ but are actually connected globally.
In order to enforce the localized cell property as in~\cite{DBLP:conf/focs/BorradaileKM08},
extra edges would need to be added for these locally disconnected components.
Instead, we enforce a \emph{global} cell property, in which
for every cell~$e$ in a cluster~$C$, there is only one (global) connected
component in the solution that intersects~$e$ and crosses the boundary of cluster~$C$.
A consequence of this is that if there are~$m$ components in the solution,
then at most~$m-1$ extra edges are needed to maintain the global cell property.
This implication is crucially used in our charging argument
to bound the cost for enforcing the cell property for the basic cells.
However, this would imply that in the dynamic program entries, we need
to keep additional information on how the portals of a cluster are connected outside the cluster.

\noindent \textbf{Combining the Ideas: A More Sophisticated Dynamic Program.}
Even though our approaches to tackle the encountered issues are intuitive,
it is a non-trivial task to balance between different tradeoffs
and keep just enough information in the dynamic program entries, but still
ensure that the entries can be filled in polynomial time.

\section{Preliminaries}
\label{sec:prelim}

We consider a metric space $M=(X,d)$ 
(see~\cite{DL97,mat-book} for more details on metric spaces).
For $x \in X$ and $\rho \geq 0$, a \emph{ball} $B(x,\rho)$ is the set $\{y \in X \mid d(x,y) \leq \rho\}$.
The \emph{diameter} $\Diam(Z)$ of a set $Z \subset X$ is the maximum distance between points in $Z$.
For $S, T \subset X$,
we denote $d(S,T) := \min\{d(x,y): x \in S, y \in T\}$,
and for $u \in X$, $d(u,T) := d(\{u\}, T)$.
Given a positive integer $m$, we denote $[m] := \{1,2,\ldots, m\}$.

A set $S \subset X$ is a $\rho$-packing, if any two distinct
points in $S$ are at a distance more than $\rho$ away from each other.
A set $S$ is a $\rho$-cover for $Z \subseteq V$,
if for any $z \in Z$, there exists $x \in S$ such that $d(x,z) \leq \rho$.
A set $S$ is a $\rho$-net for $Z$, if $S$ is a $\rho$-packing
and a $\rho$-cover for $Z$.  We assume that a $\rho$-net for any ball
in $X$ can be constructed efficiently.

We consider metric spaces with \emph{doubling dimension}~\cite{Assouad83,DBLP:conf/focs/GuptaKL03} at most $k$;
this means that
for all $x \in X$, for all $\rho> 0$, every ball $B(x,2\rho)$ can be covered
by the union of at most $2^k$ balls of the form $B(z, \rho)$, where $z
\in X$.  The following fact captures a standard property
of doubling metrics.


\begin{fact}[Packing in Doubling Metrics~\cite{DBLP:conf/focs/GuptaKL03}] \label{fact:net}
Suppose in a metric space with doubling dimension at most $k$, a $\rho$-packing $S$ has diameter at most $R$.
Then, $|S| \leq (\frac{2R}{\rho})^k$.
\end{fact}

Given an undirected graph $G = (V,E)$, where $VV \subset X$, $E \subseteq {V \choose 2}$,
and an edge $e = \{x,y\} \in E$ receives weight $d(x,y)$ from the metric space $M$.
The weight $w(G)$ or cost of a graph is the sum of its edge weights.
Let $V(G)$ denote the vertex set of a graph $G$.

\noindent We consider the \textbf{Steiner Forest Problem (SFP).}
Given a collection $\INS = \{(a_i, b_i) \mid i \in [n] \}$ of terminal pairs
in $X$, the goal is to find an undirected graph $F$ (having vertex set in $X$)
with minimum cost such that each pair of terminals are connected in $F$.
The non-terminal vertices in $V(F)$ are called Steiner points.

\noindent \textbf{Rescaling Instance.} Fix constant $\eps > 0$.
Since we consider asymptotic running time to obtain $(1 + \eps)$-approximation,
we consider sufficiently large $n > \frac{1}{\eps}$.
Suppose $R > 0$ is the maximum distance between a pair of terminals.
Then $R$ is a lower bound on the cost of an optimal solution.
Moreover, the optimal solution $F$ has cost
at most $nR$, and hence, we do not need to consider distances larger than $nR$.
Since $F$
contains at most $4n$ vertices, if we consider an $\frac{\eps R}{32n^2}$-net $S$ for $X$ and replace every point in $F$ with its closest net-point in $S$,
the cost increases by at most $\eps \cdot \OPT$.
Hence, after rescaling, we can assume that inter-point distance is at least 1 and
we consider distances up to $O(\frac{n^3}{\eps}) = \poly(n)$.
By the property of doubling dimension (Fact~\ref{fact:net}),
we can hence assume $|X| \leq O(\frac{n}{\eps})^{O(k)} \leq O(n)^{O(k)}$.

\noindent \textbf{Hierarchical Nets.} 
As in~\cite{DBLP:conf/stoc/BartalGK12},
we consider some parameter $s = (\log n)^{\frac{c}{k}} \geq 4$,
where $0<c<1$ is a universal constant that is sufficiently small (as required in Lemma~\ref{lemma:running_time}).
Set $L := O(\log_s n) = O(\frac{k \log n}{\log \log n})$. 
A greedy algorithm can construct $N_L \subseteq N_{L-1} \subseteq \cdots \subseteq N_1 \subseteq N_0 = N_{-1}= \cdots = X$
such that for each $i$, $N_i$ is an $s^i$-net for $X$,
where we say \emph{distance scale} $s^i$ is of \emph{height} $i$.

\noindent \textbf{Net-Respecting Solution.}
As defined in~\cite{DBLP:conf/stoc/BartalGK12},
a graph $F$ is net-respecting
with respect to $\{N_i\}_{i\in[L]}$ and $\enr > 0$ if for
every edge $\{x,y\}$ in $F$, both
$x$ and $y$ belong to $N_i$, where $s^i \leq \enr \cdot d(x,y) < s ^{i+1}$.

Given an instance $\INS$ of a problem,
let $\OPT(\INS)$ be an optimal solution;
when the context is clear, we also use $\OPT(\INS)$ to denote
the cost $w(\OPT(\INS))$
as well; similarly, $\OPT^{nr}(\INS)$ refers to an optimal net-respecting solution.

\subsection{Overview}
\label{sec:sparse_overview}

As in~\cite{DBLP:conf/stoc/BartalGK12,DBLP:conf/soda/ChanJ16}, we achieve a PTAS for \SFP
by the framework of sparse instance decomposition.

\noindent \textbf{Sparse Solution and Dynamic Program.} Given a graph $F$ and 
a subset $S \subseteq X$, $F|_X$ is the subgraph induced
by the vertices in $V(F) \cap X$.
A graph $F$ is called $q$-sparse, if
for all $i \in [L]$ and all $u \in N_i$,
$w(F|_{B(u, 3 s^i)}) \leq q \cdot s^i$.

We show that for \SFP (in Section~\ref{sec:ptas_sparse})
there is a dynamic program \DP that runs in
polynomial time such that
if an instance~$\INS$
has an optimal net-respecting solution that is $q$-sparse
for some small enough $q$,
$\DP(\INS)$ returns a $(1+\eps)$-approximation
with high probability (at least $1 - \frac{1}{\poly(n)}$).

\noindent \textbf{Sparsity Heuristic.} 
Since one does not know the optimal solution in advance,
we estimate the local sparsity with a heuristic.
For $i \in [L]$ and $u \in N_i$, given an instance $\INS$,
the heuristic $\heur^{(i)}_u(\INS)$ is supposed to
estimate the sparsity of an optimal net-respecting solution
in the ball $B':=B(u, O(s^i))$.
We shall see in Section~\ref{section:sparse_sfp} that
the heuristic actually gives a constant approximation
to some appropriately defined sub-instance $\INS'$
in the ball $B'$.  


\noindent \textbf{Generic Algorithm.}
We describe a generic framework that applies to \SFP.
Similar framework is also used in \cite{DBLP:conf/soda/ChanJ16, DBLP:conf/stoc/BartalGK12}
to obtain PTAS's for TSP related problems.
Given an instance $\INS$, we describe the recursive
algorithm  $\ALG(\INS)$ as follows.

\begin{compactitem}
\item[1.] \textbf{Base Case.} If $|\INS|=n$ is smaller than some constant threshold,
solve the problem by brute force, recalling that $|X| \leq O(\frac{n}{\eps})^{O(k)}$.

\item[2.] \textbf{Sparse Instance.} If for all $i \in [L]$,
	for all $u \in N_i$, $\heur^{(i)}_u(\INS)$ is at most $q_0 \cdot s^i$, for some appropriate threshold $q_0$,
call the subroutine
$\DP(\INS)$ to return a solution, and terminate.

\item[3.] \textbf{Identify Critical Instance.} Otherwise, let $i$ be the smallest height such that
there exists $u \in N_i$ with \emph{critical} $\heur^{(i)}_u(\INS) > q_0 \cdot s^i$; in this case,
choose $u \in N_i$ such that $\heur^{(i)}_u(\INS)$ is maximized.

\item[4.] \textbf{Decomposition into Sparse Instances.} Decompose the instance $\INS$ into appropriate
sub-instances $\INS_1$ and $\INS_2$ (possibly using randomness).
Loosely speaking, $\INS_1$ is a sparse enough sub-instance induced in the region around
$u$ at distance scale $s^i$, and $\INS_2$ captures the rest.
We note that $\heur^{(i)}_u(\INS_2) \leq q_0 \cdot s^i$ such that the recursion will terminate.
The union of the solutions
to the sub-instances will be a solution to $\INS$.  Moreover,
the following property holds.

\vspace{-10pt}

\begin{equation}
\expct{\OPT(\INS_1)}
\leq \frac{1}{1 - \eps} \cdot (\OPT^{nr}(\INS) - \expct{\OPT^{nr}(\INS_2)}),
\label{eq:decompose}
\end{equation}

where the expectation is over the randomness of the decomposition.

\item[5.] \textbf{Recursion.} Call the subroutine $F_1 := \DP(\INS_1)$,
and solve $F_2 := \ALG(\INS_2)$ recursively;
return the union $F_1 \cup F_2$.
\end{compactitem}

\noindent \textbf{Analysis of Approximation Ratio.}
We follow the inductive proof as in~\cite{DBLP:conf/stoc/BartalGK12}
to show that with constant probability (where
the randomness comes from $\DP$), $\ALG(\INS)$ returns a tour
with expected length at most $\frac{1+\eps}{1-\eps} \cdot \OPT^{nr}(\INS)$,
where expectation is over the randomness of decomposition into sparse instances
in Step 4.

As we shall see,
in $\ALG(\INS)$, the subroutine $\DP$ is called
at most $\poly(n)$ times (either explicitly in the recursion or 
the heuristic $\heur^{(i)}$).  Hence, with constant probability,
all solutions returned by all instances of $\DP$ have appropriate approximation
guarantees.

Suppose $F_1$ and $F_2$ are solutions returned by $\DP(\INS_1)$ and
$\ALG(\INS_2)$, respectively.
Since we assume that $\INS_1$ is sparse enough
and $\DP$ behaves correctly, $w(F_1) \leq 
(1+\eps) \cdot \OPT(\INS_1)$.
The induction hypothesis states that
$\expct{w(F_2) | \INS_2} \leq \frac{1+\eps}{1-\eps} \cdot \OPT^{nr}(\INS_2)$.

In Step 4, equation~(\ref{eq:decompose}) guarantees that
$\expct{\OPT(\INS_1)} 
\leq \frac{1}{1 - \eps} \cdot (\OPT^{nr}(\INS) - \expct{\OPT^{nr}(\INS_2)})$.
Hence, it follows that $\expct{w(F_1) + w(F_2)} \leq \frac{1+\eps}{1-\eps} \cdot \OPT^{nr}(\INS) = (1 + O(\eps))\cdot \OPT(\INS)$, achieving the desired ratio.

\noindent \textbf{Analysis of Running Time.}
As mentioned above, 
if $\heur^{(i)}_u(\INS)$ is found to be critical,
then in the decomposed sub-instances $\INS_1$ and $\INS_2$,
$\heur^{(i)}_u(\INS_2)$ should be small.
Hence, it follows that there will be at most $|X| \cdot L = \poly(n)$
recursive calls to $\ALG$.
Therefore, as far as obtaining polynomial running times,
it suffices to analyze the running time of the dynamic program $\DP$.
The details are in Section~\ref{section:dp_alg}.

\subsection{Paper Organization}
\label{sec:org}

In order to apply the above framework
to obtain a PTAS for \SFP,
we shall
describe in details the following components.

\begin{compactitem}

\item[1.]  (Section~\ref{section:sparse_sfp}.) Design a heuristic $\heur$
such that for each $i \in [L]$ and $u \in N_i$,
the heuristic $\heur^{(i)}_u(\INS)$
gives an upper bound for $\OPT^{nr}(\INS)|_{B(u,3s^i)}$.

\item[2.] (Section~\ref{section:div_n_con}.) When a critical $\heur^{(i)}_u(\INS)$ is found,
decompose $\INS$ into instances $\INS_1$ and $\INS_2$
such that equation~(\ref{eq:decompose}) holds.

\item[3.] (Section~\ref{sec:ptas_sparse}.) Design a dynamic program $\DP$
that gives $(1+\eps)$-approximation to sparse instances
in polynomial time.
\end{compactitem}

\section{Sparsity Heuristic for \SFP}
\label{section:sparse_sfp}

Suppose a collection $\INS$ of terminal pairs is an instance
of \SFP.
For $i \in [L]$ and $u \in N_i$, recall that we wish to estimate $\OPT^{nr}(\INS)|_{B(u,3s^i)}$
with some heuristic $\heur^{(i)}_u(\INS)$.  We consider
a more general heuristic $\T^{(i,t)}_u$ associated
with the ball $B(u, t s^i)$, for $t \geq 1$.  The following
auxiliary sub-instance deals with terminal pairs that are separated by
the ball.

\noindent \textbf{Auxiliary Sub-Instance.}
Fix $\delta := \Theta(\frac{\epsilon}{k})$,
where the constant depends
on the proof of Lemma~\ref{lemma:combine_cost}.
For $i \in [L]$, $u\in N_i$ and $t \geq 1$, the sub-instance
$\INS^{(i,t)}_u$ is induced by each pair $\{a,b\} \in \INS$ as follows.

\begin{compactitem}
	\item[(a)] If both $a, b \in B(u, t s^i)$, or if exactly one of them is in $B(u, t s^i)$ and the other in $B(u, (t+\delta)s^i)$,
		then $\{a, b\}$ is also included in $\INS^{(i,t)}_u$.
	\item[(b)] Suppose $j$ is the index such that 
	$s^j < \delta s^i \leq s^{j+1}$.
	If  $a \in B(u, t s^i)$ and $b \notin B(u, (t+\delta)s^i)$,
		then $\{a, a'\}$ is included in $\INS^{(i,t)}_u$, where $a'$ is the nearest point to $a$ in $N_j$.
	\item[(c)] If both $a$ and $b$ are not in $B(u, t s^i)$,
	then the pair is excluded.
\end{compactitem}

\noindent\textbf{Defining Heuristic.}
We define $\heur^{(i)}_u(\INS) := \T^{(i,4)}_u(\INS)$
in terms of a more general heuristic,
where $\T^{(i,t)}_u(\INS)$ is the cost of
a constant approximate net-respecting solution of \SFP on the instance
$\INS^{(i,t)}_u$.  For example, 
we can first apply the primal-dual algorithm in~\cite{DBLP:journals/siamcomp/GoemansW95} that gives a $2$-approximation of $\SFP$,
and then make it net-respecting and we have $\T^{(i,t)}_u(\INS) \leq 2 (1 + \Theta(\eps)) \cdot \OPT(\INS^{(i,t)}_u)$.

One potential issue is that $\OPT^{nr}(\INS)$ might
use Steiner points in $B(u, ts^i)$, even if
$\INS^{(i,t)}_u$ is empty.
We shall prove a structural property of Steiner tree in Lemma~\ref{lemma:near_terminal},
and Lemma~\ref{lemma:near_terminal} implies Lemma~\ref{lemma:long_chain}
which helps us to resolve this issue.
Recall that the Steiner tree problem is a special case
of \SFP where the goal is to return a minimum cost tree
that connects all terminals.

\begin{lemma}[Distribution of Steiner Points in The Optimal Steiner Tree]
\label{lemma:near_terminal}
Suppose $S$ is a terminal set with $\Diam(S) \leq D$,
and suppose $F$ is an optimal Steiner tree with terminal set $S$.
If the longest edge in $F$ has weight at most $\gamma D$ ($0<\gamma \leq 1$),
then for any Steiner point $r$ in $F$, $d(r, S) \leq  4 k \gamma \log_2 \frac{4}{\gamma} \cdot D$.
\end{lemma}
\begin{proof}
Since $F$ is an optimal solution, all Steiner points in $F$ have degree at least $3$.
Fix any Steiner point $r$ in $F$.

Denote $K := \ceil{\log_{2}(\gamma D)}$.  Suppose we consider $r$ as the root of the tree $F$.
We shall show
that there is a path of small weight from $r$ to some terminal.
Without loss of generality, we can assume that all terminals are leaves,
because once we reach a terminal, there is no need to visit its descendants.
For simplicity, we can assume that each internal node (Steiner point)
has exactly two children, because we can ignore extra branches if an internal
has more than two children.

For $i \leq K$, let $E_i$ be the set of edges in $F$ 
that have weights in the range $(2^{i-1}, 2^i]$,
and we say that such an edge is of \emph{type $i$}.
For each node $u$ in $F$,
denote $F_u$ as the subtree rooted at $u$.
Suppose we consider $F_u$ and remove all edges
in $\cup_{j \geq i} E_j$
from $F_u$; in the resulting forest,
let $M^{(i)}_u$ be the number of connected components
that contain at least one terminal.
We shall prove the following statement by structural induction
on the tree $\widehat{F}$.

\emph{For each node $u \in F$,
there exists a leaf $x \in F_u$
such that $d(x,u) \leq \sum_{i \leq K} 2^i \log_2 M^{(i)}_u$.}

\noindent \textbf{Base Case.} If $u$ is a leaf,
then the statement is true.

\noindent \textbf{Inductive Step.}  Suppose $u$ has children $u_1$ and $u_2$
such that $\{u, u_1\} \in E_i$ and $\{u, u_2\} \in E_{i'}$,
where $i \geq i'$.
Suppose $x_1$ and $x_2$ are the leaves in $F_{u_1}$ and $F_{u_2}$,
respectively, from the induction hypothesis.  
Observe that $M^{(i)}_u = M^{(i)}_{u_1} + M^{(i)}_{u_2}$.
We consider two cases.

\noindent (1) Suppose $M^{(i)}_{u_1} \leq M^{(i)}_{u_2}$.
Then, we can pick $x_1$ to be the desired leaf,
because the extra distance $d(u_1, u) \leq 2^i$
can be accounted for, as
$2 M^{(i)}_{u_1} \leq M^{(i)}_u$,
and $M^{(j)}_{u_1} \leq M^{(j)}_{u}$ for $j \neq i$.
More precisely,
$d(x_1, u) \leq d(x_1, u_1) + d(u_1, u)
\leq 2^i \cdot (1 + \log_2 M^{(i)}_{u_1})
+ \sum_{j \leq K: j \neq i} 2^j \log_2 M^{(j)}_{u_1}
\leq \sum_{j \leq K} 2^j \log_2 M^{(j)}_{u}$,
where the second inequality follows from the induction
hypothesis for $u_1$.

\noindent (2) Suppose $M^{(i)}_{u_2} < M^{(i)}_{u_1}$.
Then, similarly we pick $x_2$ to be the desired leaf,
because the extra distance is $d(u_2, u) \leq 2^{i'} \leq 2^i$.
This completes the inductive step.

Next, it suffices to give an upper bound for
each $M^{(i)} := M^{(i)}_r$ for root $r$.
Suppose after removing all tree edges in $\cup_{j \geq i} E_j$,
$P$ and $Q$ are two clusters each containing at least one terminal.
Then, observe that the path in $F$ connecting $P$ and $Q$
must contain an edge $e$ with weight at least $2^{i-1}$.
It follows that $d(P,Q) \geq 2^{i-1}$;
otherwise, we can replace $e$ in $F$
with another edge of length less than $2^{i-1}$ to obtain a Steiner tree with strictly less weight.
It follows that each cluster has a terminal representative
that form a $2^{i-1}$-packing.
Hence, we have $M^{(i)} \leq (\frac{4D}{2^i})^k$,
by the packing property of doubling metrics (Fact~\ref{fact:net}).

Therefore, every Steiner point $r$ in $\widehat{F}$
has a terminal within distance
$\sum_{i \leq K} k \cdot 2^i \log_2 \frac{4D}{2^i} \leq 4 k \gamma D \log_2 \frac{4}{\gamma}$.
\end{proof}

Given a graph $F$, a \emph{chain} in $F$
is specified by a sequence of points
$(p_1, p_2,\ldots, p_l)$ such that
there is an edge $\{p_i, p_{i+1}\}$ in $F$ between adjacent points,
and the degree of an internal point $p_i$ (where $2 \leq i \leq l-1$)
in $F$ is exactly 2.


\begin{lemma}[Steiner Tree of Well-Separated Terminals Contains A Long Chain]
	\label{lemma:long_chain}
Suppose $S$ and~$T$ are terminal sets in a metric space
with doubling dimension at most $k$ such that $\Diam(S\cup T) \leq D$, and
$d(S, T) \geq \tau D$, where $0<\tau<1$.
Suppose $F$ is an optimal net-respecting Steiner tree 
connecting the points in $S \cup T$.
Then, there is a chain in $F$ with weight at least $\frac{\tau^2}{4096 k^2} \cdot D$ such that
any internal point in the chain is a Steiner point.
\end{lemma}
\begin{proof}
Denote $\gamma := \frac{\tau^2}{4096 k^2}$.
Suppose for contradiction's sake that all chains in $F$ have weight less than $\gamma D$. 
We consider a minor $\widehat{F}$ that is obtained from $F$
by merging Steiner points of degree 2 with adjacent points.
Hence, the vertex set of $\widehat{F}$ are the terminals together
with Steiner points in $F$ with degree at least 3.
Moreover, an edge in $\widehat{F}$ corresponds to a chain
in $F$, and its weight is defined to be the weight
of the corresponding chain.

Then by using the argument in Lemma~\ref{lemma:near_terminal},
We can prove that every point $u$ in $\widehat{F}$ is within distance at most
$4k\gamma\log_{2}{\frac{8}{\gamma}} \cdot D$ to a terminal.
Precisely, we shall replace the $F$ in the argument of Lemma~\ref{lemma:near_terminal} with $\widehat{F}$.
We observe that the only difference caused by this replacement is when we use the optimality of the solutions.
Specifically, in Lemma~\ref{lemma:near_terminal}
we use the fact that when an edge $e$ connects point
sets $P$ and $Q$ that both contain at least one terminal (i.e. removing $e$ results in the dis-connectivity of $P$ and $Q$),
it has to be $d(P, Q) \geq w(e)$, while
the corresponding fact for $\widehat{F}$ is $d(P, Q) \geq \frac{w(e)}{1+\Theta(\epsilon)} \geq \frac{w(e)}{2}$
because of the net-respecting property.

\noindent \textbf{Obtaining Contradiction.}
Recall that the terminal sets $S$ and $T$ are well-separated $d(S,T) \geq \tau D$.
Since all Steiner points in $\widehat{F}$ are at distance at most $4k\gamma\log_{2}{\frac{8}{\gamma}} \cdot D$
from the terminals,
it follows that there must be an edge in $\widehat{F}$ with length
at least $\tau D - 8k\gamma\log_{2}{\frac{8}{\gamma}}\cdot D > \tau D - 32k\sqrt{\gamma} D >
\gamma D$.
\end{proof}

\begin{lemma}
\label{lemma:heuristic_small}
Suppose $F$ is an optimal net-respecting solution for an \SFP instance $\INS$.
Then, for any $i$ and $u\in N_i$ and $t \geq 1$, $w(F|_{B(u, t s^i)}) \leq  \T^{(i,t+1)}_u(\INS) + O(\frac{s k t}{\epsilon})^{O(k)} s^i$.
\end{lemma}

\begin{proof}
Given an optimal net-respecting solution $F$,
we shall construct another net-respecting solution
in the following steps.

\begin{compactitem}
\item[1.] Remove edges in $F|_{B(u, ts^i)}$.

\item[2.] Add edges corresponding to the heuristic
$\T^{(i,t+1)}_u(\INS)$.

\item[3.] Add edges in a minimum spanning tree $H$
of $N_j \cap B(u, (t+2)s^i)$,
where $s^j \leq \Theta(\frac{\eps}{(t+1)k^2}) \cdot
s^i < s^{j+1}$, where
the constant in Theta depends on Lemma~\ref{lemma:long_chain}; convert each added
edge into a net-respecting path if necessary.
Observe that the weight of edges
added in this step is $O(\frac{s t k}{\eps})^{O(k)} \cdot s^i$.

\item[4.] To ensure feasibility, replace some edges without increasing the weight.
\end{compactitem}

If we can show that the resulting solution is feasible 
for $\INS$, then the optimality of
$F$ implies the result.
We denote $B := B(u, ts^i)$ and $\widehat{B} := B(u, (t+1)s^i)$.

\noindent \textbf{Feasibility.}  

\noindent Define $\widehat{V}_1 := \{ x : x\in B \mid \exists \{x, y\} \in F \text{ s.t. } y\notin B 
\text{ and } y \text{ is connected in } F|_{X\setminus B}  \text{ to some point outside } \widehat{B}\}$,
and 

\noindent $\widehat{V}_2 := \{ x : x\in \widehat{B} \setminus B \mid x \text{ is connected in } F|_{\widehat{B}} \text{ to some point in } \widehat{V}_1,  \exists \{x, y\} \in F \text{ s.t. } y\notin \widehat{B}\}$.
In Step 4, we will ensure that all points in $\widehat{V}_1 \cup \widehat{V}_2$ are connected to the MST $H$.

If a pair $\{a,b\} \in \INS$ has
both terminals in $\widehat{B}$, then they will be connected
by the edges corresponding to $\T^{(i,t+1)}_u(\INS)$.  
If $a \in \widehat{B}$ and $b \notin \widehat{B}$,
then edges for the heuristic $\T^{(i,t+1)}_u(\INS)$ ensures
that $a$ is connected to $H$; moreover,
in the original tree $F$, 
if the path from $a$ to $b$ does not meet any node in 
$\widehat{V}_2$, then this path is preserved,
otherwise there is a portion of the path
from a point in $\widehat{V}_2$
to $b$ that is still preserved.  If both $a$ and $b$ are outside $\widehat{B}$,
then they might be connected in $F$ via points in $\widehat{V}_2$; however,
since all points in $\widehat{V}_2$ are connected to $H$,
feasibility is ensured.

We next elaborate how Step 4 is performed.
Consider a connected component $U$ in $F|_{\widehat{V}_1 \cup (\widehat{B}\setminus B)}$
that contains a point in $\widehat{V}_1$.  Let $S_1 := U \cap \widehat{V}_1$
and $S_2 := U \cap \widehat{V}_2$.  If $S_2 = \emptyset$,
then there is an edge connecting $S_1$ directly to a point outside $\widehat{B}$.
This means that both its end-points are in $N_j$ by the net-respecting property, and hence
$S_1$ is already connected to $H$.

Next, if there is a point $z \notin \widehat{B}$ connected directly
to some point $y \in S_2$ such that $d(y,z) \geq \frac{s^i}{2}$,
then by the net-respecting property, $y \in N_j$ and so again $U$ is connected to $H$.
Otherwise, we have $d(S_1, S_2) \geq \frac{s^i}{2}$.  We next replace $U$
with an optimal net-respecting Steiner tree $\widehat{U}$ connecting $S_1 \cup S_2$.  Since $U$ itself is net-respecting, this does not increase the cost.

Observing that $\Diam(S_1 \cup S_2) \leq 2(t+1)s^i$,
we can use Lemma~\ref{lemma:long_chain} to conclude that
there exists a chain in $\widehat{U}$ from some point $u$ to $v$
such that its length is at least $\Theta(\frac{1}{k^2 (t+1)}) \cdot s^i$.
Hence, we can remove this chain, and use its weight to add a net-respecting path
from each of $u$ and $v$ to its nearest point in $N_j$.  This does not increase the cost, and ensures that both $S_1$ and $S_2$ are connected to $H$.

Therefore, we have shown that Step 4 ensures that all points in $\widehat{V_1}$ and 
$\widehat{V_2}$ are connected to $H$.
\end{proof}

It is because of Lemma~\ref{lemma:heuristic_small} that we choose  
$\heur^{(i)}_u(\INS) := \T^{(i,4)}_u(\INS)$ to be the heuristic.

\begin{corollary}[Threshold for Critical Instance]
\label{cor:threshold}
	Suppose $F$ is an optimal net-respecting solution for an \SFP instance $\INS$,
	and $q \geq \Theta(\frac{s k}{\epsilon})^{\Theta(k)}$.
	If for all $i\in [L]$ and $u\in N_i$, $\heur^{(i)}_u(\INS) \leq q s^i$, then $F$ is $2q$-sparse.
\end{corollary}

\begin{lemma}
\label{lemma:heuristic_upper_bound}
	Suppose $\INS$ is an \SFP instance. Consider $i\in [L]$, $u\in N_i$, and $t\geq t'\geq 1$. Suppose $F$ is a
	net-respecting solution for $\INS^{(i,t)}_u$.
	Then, $\T^{(i,t')}_u(\INS) \leq 4 (1 + \eps) \cdot w(F) + O(\frac{s k t'}{\epsilon})^{O(k)} s^i$.
\end{lemma}
\begin{proof}
	We first show that there is a feasible solution for  $\INS^{(i,t')}_u$ with weight
	at most $2 \cdot w(F) + O(\frac{s k t'}{\epsilon})^{O(k)} s^i$.
	Then, the heuristic $\T^{(i,t')}_u(\INS)$ gives the weight of
	a net-respecting solution with cost at most
	$4 (1 + \eps) \cdot w(F) + O(\frac{s k t'}{\epsilon})^{O(k)} s^i$.
	
	We first include $F$ in the solution.  It suffices
	to handle the terminal pairs in  $\INS^{(i,t')}_u \setminus \INS^{(i,t)}_u$.
	Such a pair $\{a, a'\}$ must be induced from $\{a,b\} \in \INS^{(i,t)}_u$
	such that $a \in B(u, t' s^i)$ and $b \in B(u, (t + \delta)s^i) \setminus
	B(u, (t' + \delta)s^i)$.  We next add more edges such that $a$ is connected to $a'$, which lies in
  $N_j$, where $s^j \leq \Theta(\frac{\delta^2}{t' k^2}) \cdot s^i < s^{j+1}$.  
	
	We add a minimum spanning tree $H$ on the points in $N_j \cap B(u, (t' + \delta) s^i)$.
	This has cost at most $O(\frac{s k t'}{\epsilon})^{O(k)} \cdot s^i$.
	
	Consider a connected component $U$ of $F$.  Consider the
	terminal pairs $\{a,b\} \in \INS^{(i,t)}_u$ connected by $U$
	such that $a \in B(u, t' s^i)$ and $b \in B(u, (t + \delta)s^i)$;
	let $S_1$ be those terminals $a$'s, and $S_2$ be those terminal $b$'s.
	Suppose $\widehat{U}$ is an optimal net-respecting Steiner tree
	connecting $S_1 \cup S_2$.  Since $U$ is also net-respecting,
	it follows that the weight of $\widehat{U}$ is at most that of $U$.
	
	Since $d(S_1, S_2) \geq \delta s^i$ and $\Diam(S_1 \cup S_2) \leq \Theta(t') s^i$,
	it follows from Lemma~\ref{lemma:long_chain} that
	there exists a chain from $p$ to $q$ in $\widehat{U}$ with
	weight at least $\Theta(\frac{\delta^2}{t' k^2}) \cdot s^i$.
	Hence, we can remove this chain, and use this weight to connect $p$ and $q$
	to each of their closest points in $N_j$.
	This ensures that each point $a \in S_1$ is connected to its closest point in $N_j$ via the minimum spanning tree $H$.
		
	If we perform this operation on each connected component $U$ of $F$,
	the weight of edges added is at most $w(F)$.  Hence,
	we have shown that there is a feasible solution
	to $\INS^{(i,t')}_u$ with cost 
	at most $2 \cdot w(F) + O(\frac{s k t'}{\epsilon})^{O(k)} s^i$,
	as required.
\end{proof}

\section{Decomposition into Sparse Instances}
\label{section:div_n_con}

In Section~\ref{section:sparse_sfp},
we define a heuristic $\heur^{(i)}_u(\INS)$
to detect a critical instance around some point $u \in N_i$
at distance scale $s^i$.  We next describe
how the instance $\INS$ can be decomposed into
$\INS_1$ and $\INS_2$ such that equation~(\ref{eq:decompose})
in Section~\ref{sec:sparse_overview}
is satisfied.

Since the ball centered at $u$
with radius around $s^i$ could potentially
separate terminal pairs in $\INS$,
we use the idea in Section~\ref{section:sparse_sfp}
for defining the heuristic to decompose the
instance.

\noindent \textbf{Decomposing a Critical Instance.}
We define a threshold $q_0 := \Theta(\frac{s k}{\epsilon})^{\Theta(k)}$ according to 
Corollary~\ref{cor:threshold}.  As stated
in Section~\ref{sec:sparse_overview},
a critical instance is detected by the heuristic
when a smallest $i \in [L]$ is found
for which there exists some $u \in N_i$
such that $\heur^{(i)}_u(\INS) = \T^{(i,4)}_u(\INS) > q_0 s^i$.
Moreover, in this case, $u \in N_i$ is chosen
to maximize $\heur^{(i)}_u(\INS)$.
To achieve a running time with an
$\exp(O(1)^{k \log(k)})$ dependence
on the doubling dimension $k$,
we also apply the technique in~\cite{DBLP:conf/soda/ChanJ16} to
choose the cutting radius carefully.

\begin{claim}[Choosing Radius of Cutting Ball]
\label{claim:cut_ball}
Denote $\T(\lambda) := \T^{(i, 4 + 2\lambda)}_u(\INS)$.
Then, there exists $0 \leq \lambda < k$ such that
$\T(\lambda+1) \leq 30 k \cdot \T(\lambda)$.
\end{claim}

\begin{proof}
Suppose the contrary is true.
Then, it follows that $\T(k) > (30k)^k \cdot \T(0)$.
We shall obtain a contradiction by showing that 
there is a solution for the instance $\INS^{(i,4+2k)}_u$ corresponding to $\T(k) = \T^{(i,4+2k)}_u(\INS)$ with small weight.

Define $N_{i}'$ to be the set of points in $N_i$ that cover $B(u, (2k+5)s^i)$, and similarly define $N_{j}'$, where $s^j \leq \delta \cdot s^i \leq s^{j+1}$.

Define edge set $F$ to be the union of a minimum spanning tree on $N_j'$
together with the union of the edge sets $\heur^{(i)}_v$ over $v \in N_i'$.
It follows that $F$ is a feasible solution for the instance
$\INS^{(i,4+2k)}_u$.
By the choice of $u$ and $q_0$,
we have $w(F) \leq  |N_j'| \cdot 2(2k+5) \cdot s^i + |N_i'| \cdot \T(0)
\leq q_0 s^i + (4k+10)^k \cdot \T(0) \leq (15k)^k \cdot \T(0)$.

Hence, we have an upper bound
for the heuristic $\T(k) \leq 2(1 + \Theta(\eps)) \cdot w(F) \leq (30)^k \cdot \T(0)$,
which gives us the desired contradiction.
\end{proof}

\noindent \textbf{Cutting Ball and Sub-Instances.}
Suppose $\lambda \geq 0$ is picked as in Claim~\ref{claim:cut_ball},
and sample $h \in [0, \frac{1}{2}]$ uniformly at random.
Recall that $\delta := \Theta(\frac{\eps}{k})$.
Define $B := B(u, (4 + 2 \lambda + h)s^i)$ and
$\widehat{B} := B(u, (4 + 2 \lambda + h + \delta)s^i)$.
The instances $\INS_1$ and $\INS_2$ are induced by each
pair $\{a,b\} \in \INS$ as follows.

\begin{compactitem}
\item[(a)]  If $a \in B$ and $b \in \widehat{B}$, then include $\{a, b\}$ in $\INS_1$.

\item[(b)] If $a \in B$ and $b \notin \widehat{B}$,
then include $\{a, a'\}$ in $\INS_1$ and $\{a',b\}$ in $\INS_2$,
where $a'$ is the closest point in $N_j$ to $a$ and $s^j \leq \delta \cdot s^i < s^{j+1}$.

\item[(c)] If both $a$ and $b$ are not in $B$, then
include $\{a, b\}$ in $\INS_2$.
\end{compactitem}

\begin{lemma}[Sub-Instances Are Sparse]
\label{lemma:subinstance}
The sub-instances $\INS_1$ and $\INS_2$ satisfy the following.
\begin{compactitem}
\item[(i)] If $F_1$ is feasible for $\INS_1$ and $F_2$ is feasible for $\INS_2$,
then the union $F_1 \cup F_2$ is feasible for $\INS$.
\item[(ii)] The sub-instance $\INS_2$ does not have a critical instance
with height less than $i$, and $\heur^{(i)}_u(\INS_2) = 0$.
\item[(iii)] $\heur^{(i)}_u(\INS_1) \leq O(s)^{O(k)} \cdot q_0 \cdot s^i$.
\end{compactitem}
\end{lemma}

\begin{proof}
The first two statements follow immediately from the construction.  For the third statement,
we use the fact that there is no critical instance at height $i-1$
to show that there is a solution to $\INS_1$ with small cost.

Specifically, we consider a minimum spanning tree $H$
on $N_j \cap B(u, 5 s^i)$, where $s^j \leq \delta \cdot s^{i-1} < s^{j+1}$.
Then, we have $w(H) \leq q_0 \cdot s^i$.

Moreover, we consider the union of solutions corresponding
to $\heur^{(i-1)}_v(\INS)$, over $v \in N_{i-1} \cap B(u, 5s^i)$.
The cost is $O(s)^{O(k)} \cdot q_0 \cdot s^i$.

Hence, the union of $H$ together with the edges for the $\heur^{(i-1)}_v(\INS)$'s
is feasible for $\INS_1$, and this implies that
$\heur^{(i)}_u(\INS_1) \leq O(s)^{O(k)} \cdot q_0 \cdot s^i$.
\end{proof}





\begin{lemma}[Combining Costs of Sub-Instances]
\label{lemma:combine_cost}
	Suppose $F$ is an optimal net-respecting solution for $\INS$.
	Then, for any realization of the decomposed sub-instances $\INS_1$ and $\INS_2$
	as described above, there exist net-respecting solutions $F_1$ and $F_2$
	for $\INS_1$ and $\INS_2$, respectively, such that
	$(1-\epsilon) \cdot \expct{w(F_1)} + \expct{w(F_2)} \leq w(F)$, where the expectation
	is over the randomness to generate $\INS_1$ and $\INS_2$.
\end{lemma}

\begin{proof}
Let $B$ and $\widehat{B}$ be defined as above,
and denote $\overline{B} := B(u, (4 + 2 \lambda + 1) \cdot s^i)$.
Hence, $B \subset \widehat{B} \subset \overline{B}$.

We start by including $F|_{B}$ in $T_1$,
and including the remaining edges in $F$ in $F_2$. We will then show how to add extra edges with expected weight
at most $\epsilon \cdot \expct{w(F_1)}$  to make $F_1$ and $F_2$ feasible.
This will imply the lemma.

Define $N$ to be the subset of $N_j$ that cover the points in $\overline{B}$,
where $s^j < \delta s^i \leq s^{j+1}$.
We include a copy of a minimum spanning tree $H$ of $N$ in each of $F_1$ and $F_2$, and make it net-respecting.
This costs at most $|N|\cdot O(k) \cdot s^i \leq O(\frac{k s}{\epsilon})^{O(k)} \cdot s^i$.

We next include the edges of $F$ in the annulus
$\widehat{B} \setminus B$ (of width $\delta$) into $F_1$.
This has expected cost at most $\delta \cdot w(F|_{\overline{B}})$.

\noindent\textbf{Connecting Crossing Points.}
To ensure the feasibility of $F_1$, we
connect the following sets of points to $N$.
We denote:

$V_1 := \{x \in B \mid \exists y \in \widehat{B} \setminus B, 
\{x,y\} \in F\}$,
$V_2 := \{y \in \widehat{B} \setminus B \mid
\exists x \in B, \{x,y\} \in F\}$, and

$V_3 := \{x \in \widehat{B} \mid 
\exists y \notin \widehat{B}, 
\{x,y\} \in F\}$.

We shall connect each point in $V_1 \cup V_2 \cup V_3$
to its closest point in $N$.  Note that if
such a point $x$ is incident to some edge in $F$
with weight at least $\frac{s^i}{4}$,
then the net-respecting property of $F$ implies
that $x$ is already in $N$.  Otherwise,
this is because some edge $\{x,y\}$ in $F$
is cut by either $B$ or $\widehat{B}$, which
happens with probability at most $O(\frac{d(x,y)}{s^i})$.
Hence, each edge $\{x,y\} \in F|_{\overline{B}}$
has an expected contribution of $\delta s^i \cdot O(\frac{d(x,y)}{s^i}) = O(\delta) \cdot d(x,y)$.

Similarly, to ensure the feasibility of $F_2$,
we ensure each point in the following set is
connected to $N$.  Denote
$\widehat{V_1} := \{x \in {B} \mid 
\exists y \notin {B}, 
\{x,y\} \in F\}$.  By the same argument,
the expected cost to connect each point to $N$ is also at most $O(\delta) \cdot w(F|_{\overline{B}})$.

\noindent\textbf{Charging the Extra Costs to $F_1$.}
Apart from using edges in $F$,
the extra edges come from two copies of the
minimum spanning tree $H$,
and other edges with cost $O(\delta) \cdot w(F|_{\overline{B}})$.  We charge these extra costs
to $F_1$.

Since $T^{(i,4)}_u(\INS) > q_0 \cdot s^i$ and $F_1$ is a net-respecting solution for $\INS^{(i,4+2\lambda+h)}_u$,
by Lemma~\ref{lemma:heuristic_upper_bound},
$w(F_1) \geq \frac{1}{4(1+\eps)}(T^{(i)}(u, 4) - O(\frac{s k}{\epsilon})^{O(k)} \cdot s^i) > \frac{q_0}{8} \cdot s^i$,
by choosing large enough $q_0$.

Therefore, the cost for the two copies of the
minimum spanning tree $H$ is at most $O(\frac{k s}{\epsilon})^{O(k)} \cdot s^i
\leq \frac{\epsilon}{2} \cdot w(F_1)$.

We next give an upper bound on $w(F|_{\overline{B}})$,
which is at most 
$\T^{(i,4+2(\lambda+1))}_u(\INS) + O(\frac{s k}{\epsilon})^{O(k)}
\cdot s^i$, by Lemma~\ref{lemma:heuristic_small}.
By the choice of $\lambda$,
we have $\T^{(i,4+2(\lambda+1))}_u(\INS) \leq 30k \cdot 
\T^{(i,4+2\lambda + 1)}_u(\INS)$.
Moreover, 
by Lemma~\ref{lemma:heuristic_upper_bound},
$\T^{(i,4+2\lambda + 1)}_u(\INS) \leq 4(1+\eps) \cdot w(F_1) +
O(\frac{sk}{\eps})^{O(k)} \cdot s^i$.
Hence, we can conclude that
$w(F|_{\overline{B}}) \leq O(k) \cdot w(F_1)$.

Hence, by choosing small enough $\delta = \Theta(\frac{\eps}{k})$,
we can conclude that the extra
costs $O(\delta) \cdot w(F|_{\overline{B}})
\leq \frac{\eps}{2} \cdot w(F_1)$.

Therefore, we have shown that
$\expct{w(F_1)} + \expct{w(F_2)}
\leq w(F) + \eps \cdot w(F_1)$, where the right hand side is a random variable.
Taking expectation on both sides and rearranging gives
the required result.
\end{proof}

\section{A PTAS for Sparse $\SFP$ Instances}
\label{sec:ptas_sparse}

Our dynamic program follows the divide and conquer strategy
as in previous works on \TSP~\cite{DBLP:journals/jacm/Arora98,DBLP:conf/stoc/Talwar04,DBLP:conf/stoc/BartalGK12}
that are based on hierarchical decomposition.
However, to apply the framework to \SFP, we need a version of the \emph{cell property}
that is more sophisticated than previous works~\cite{DBLP:conf/focs/BorradaileKM08,DBLP:journals/algorithmica/BateniH12}.

We shall first give a review of the hierarchical decomposition techniques in Section~\ref{sec:hier_decomp}.
Then in Section~\ref{section:dp_struct},
we shall define our cell property precisely, and also prove that there exist good solutions that satisfy the cell property
(in Lemma~\ref{lemma:struct_property}).
Finally, we shall define \DP in Section~\ref{section:dp_alg}, and conclude a PTAS for sparse $\SFP$ instances
(in Corollary~\ref{corollary:ptas_sparse_sfp}).

\subsection{Review on Hierarchical Decomposition}
\label{sec:hier_decomp}

\begin{definition}[Single-Scale Decomposition~\cite{DBLP:conf/stoc/AbrahamBN06}]
\label{defn:single_decomp}
At height $i$, an arbitrary ordering $\pi_i$ is
imposed on the net $N_i$.  Each net-point $u \in N_i$
corresponds to a \emph{cluster center} and
samples random $h_u$ from a truncated exponential distribution
$\Exp_i$ having density function $t \mapsto
\frac{{\chi}}{{\chi}-1} \cdot \frac{\ln \chi}{s^i} \cdot e^{-\frac{t \ln \chi }{s^i}}$ for $t \in [0, s^i]$, where $\chi = O(1)^k$.
Then, the cluster at $u$ has random radius $r_u := s^i + h_u$.

The clusters
induced by $N_i$ and the random radii form a
decomposition $\Pi_i$,
where a point $p \in V$ belongs to the cluster
with center $u \in N_i$ such that $u$ is the first
point in $\pi_i$ to satisfy $p \in B(u, r_u)$.
We say that the partition $\Pi_i$ cuts a set $P$
if $P$ is not totally contained within a single cluster.

The results in~\cite{DBLP:conf/stoc/AbrahamBN06} imply that
the probability that a set $P$ is cut by $\Pi_i$
is at most $\frac{\beta \cdot \Diam(P)}{s^i}$,
where $\beta = O(k)$.
\end{definition}

\begin{definition}[Hierarchical Decomposition] \label{defn:phd}
Given a configuration of random radii
for $\{N_i\}_{i \in [L]}$, decompositions $\{\Pi_i\}_{i \in [L]}$
are induced as in Definition~\ref{defn:single_decomp}.
At the top height $L-1$, the whole space
is partitioned by $\Pi_{L-1}$ to form height-$(L-1)$ clusters.  Inductively,
each cluster at height $i+1$ is partitioned
by $\Pi_i$ to form height-$i$ clusters, until height $0$ is reached.  Observe that a cluster
has $K := O(s)^k$ child clusters.
Hence, a set $P$ is cut at height $i$ \emph{iff}
the set $P$ is cut by some partition $\Pi_j$ such that
$j \geq i$; this happens with probability
at most $\sum_{j \geq i} \frac{\beta \cdot \Diam(P)}{s^i} = \frac{O(k) \cdot \Diam(P)}{s^i}$.
\end{definition}

\noindent \textbf{Portals.}
As in~\cite{Arora-tspsurvey,DBLP:conf/stoc/Talwar04,DBLP:conf/stoc/BartalGK12},
each height-$i$ cluster $U$ is equipped with portals
such that a solution $F$ is \emph{portal-respecting}, if
for every edge $\{x,y\}$  in $F$ between a point $x$ in $U$
and some point $y$ outside $U$, at least one of $x$ and $y$ must be a portal of cluster $U$.
As mentioned in~\cite{DBLP:conf/stoc/BartalGK12},
the portals of a cluster need not be points
of the cluster itself, but are just used as connection points.
For a height-$i$ cluster $C$, its portals
is the subset of net-points
in $N_{i'}$ that cover $C$,
where $i'$ is the maximum index such that
$s^{i'} \leq \max\{1, \frac{\ep}{4 \beta L} \cdot s^i\}$.
As noted in~\cite{DBLP:conf/stoc/Talwar04,DBLP:conf/stoc/BartalGK12,DBLP:conf/soda/ChanJ16},
any solution can be made to be portal-respecting
with a multiplicative factor of $1 + O(\eps)$ in cost.

Since a height-$i$ cluster
has diameter $O(s^i)$, by Fact~\ref{fact:net},
the cluster has at most $m := O(\frac{\beta L s}{\ep})^k$ 
portals.

\noindent \textbf{$(m,r)$-Light Solution.}
A solution $F$ is called $(m,r)$-\emph{light}, 
if it is portal-respecting for a hierarchical
decomposition in which each cluster has at most $m$ portals,
and for each cluster, at most $r$ of its portals
are used in $F$ to connect points in the cluster to the points outside.


\subsection{Structural Property}
\label{section:dp_struct}
In this section, we shall define the cell property (Definition~\ref{definition:cell_property}) with respect to the effective cells
(Definition~\ref{definition:eff}),
where the effective cells are carefully chosen to
implement our adaptive cells idea which is discussed in Section~\ref{sec:intro}. Specifically,
the effective cells are defined by the union of
the basic cells (Definition~\ref{definition:bas}) and the non-basic cells (Definition~\ref{definition:nbas}).
Moreover, the virtual cells and the promoted cells (Definition~\ref{definition:pro_vir}) are introduced
in order to define the non-basic cells.
Finally, we shall prove the structural property in Lemma~\ref{lemma:struct_property}.

\noindent\textbf{Notations and Parameters.}
Let $\Ht(C)$ denote the height of a cluster $C$,
$\Des(C)$ denote the collection of all descendant clusters of $C$ (including $C$),
and $\Par(C)$ denote the parent cluster of $C$.
For $x \in \mathbb{R}_{+}$, let $\lfloor x \rfloor_s$ denote the largest power of $s$ that is at most $x$,
and $\lceil x \rceil_s$ denote the smallest power of $s$ that is at least $x$.
Define $\hat{\gamma}_{0} := \Theta(\frac{\epsilon}{k s^2 L})$, and define $\hat{\gamma}_{1} := \Theta(\frac{\epsilon}{s^2})$.
Define $\gamma_0$ such that $\frac{1}{\gamma_0} := \lceil \frac{1}{\hat{\gamma}_{0}} \rceil_s$,
and define $\gamma_1$ such that $\frac{1}{\gamma_1} := \lfloor \frac{1}{\hat{\gamma}_1} \rfloor_s$.
We note that $\gamma_0 < \gamma_1$.

\begin{definition}[Cell]
	Suppose $C$ is a cluster of height $i$. A $p$-cell of $C$ is a height-$\log_{s}{p}$ sub-cluster of $C$.
\end{definition}

\begin{definition}[Crossing Component]
Suppose $C$ is some cluster, and $F$ is a solution for $\SFP$.
We say that a subset $A$ crosses $C$, if there exists points $x, y \in A$ such that $x \in C$ and $y \notin C$.
A component $A$ in $F$ is called a crossing component of $C$ if $A$ crosses $C$.
\end{definition}

In the following, we shall introduce the notions of the basic cells, owner of basic cells,
promoted cells, virtual cells, 
non-basic cells, and effective cells. All of these are defined with respect to some feasible solution to 
$\SFP$. We assume there is an underlying feasible solution $F$
when talking about these definitions.

\noindent \textbf{Adaptive Cells.}  For each cluster $C$, we shall define its \emph{basic cells}
whose heights depend on the weights $l$ of the crossing components of $C$ in the solution~$F$.
We consider three cases.

Define $I_{1}(l) := \{ i \mid \lfloor l \rfloor_s \geq s^i \}$,
$I_{2}(l) := \{i \mid \frac{\gamma_0}{\gamma_1} s^i \leq \lfloor l \rfloor_s < s^i\}$
and $I_{3}(l) := \{i \mid i\leq L, \lfloor l \rfloor_s < \frac{\gamma_0}{\gamma_1} s^i\}$.
Define a function $h : [L] \times \mathbb{R}_{+} \rightarrow \mathbb{R}_{+}$, such that
\begin{equation*}
	h(i, l)=
	\begin{cases}
		\gamma_1 s^i, & \text{ for } i\in I_{1}(l) \\
		\gamma_1 \lfloor l \rfloor_s, & \text{ for } i\in I_{2}(l) \\
		\gamma_0 s^i, & \text{ for } i \in I_{3}(l)
	\end{cases}
\end{equation*}

\begin{lemma}
	\label{lemma:h_increase_one}
	$\frac{h(i+1, l)}{s} \leq h(i, l) \leq h(i+1, l)$.
\end{lemma}
\begin{proof}
	If both $i$ and $i+1$ lie in the same $I_{j}(l)$ ($j \in \{ 1, 2, 3 \}$), then it holds immediately.

	Otherwise, it is either $i \in I_{2}(l)$ but $i+1 \in I_{3}(l)$, or $i \in I_{1}(l)$ but $i+1 \in I_{2}(l)$.
	\begin{compactitem}
		\item If $i \in I_{1}(l)$ and $i+1 \in I_{2}(l)$. This implies $s^i = \lfloor l \rfloor_s$.
			Hence, $h(i, l) = \gamma_1 s^i = \gamma_1 \lfloor l \rfloor_s = h(i+1, l)$.
		\item If $i \in I_{2}(l)$ and $i+1 \in I_{3}(l)$. This implies $s^i = \frac{\gamma_1}{\gamma_0} \lfloor l \rfloor_s$.
			Hence, $s\cdot h(i, l) = s\cdot \gamma_1 \lfloor l \rfloor_s = \gamma_0 s^{i+1} = h(i+1, l)$.
	\end{compactitem}
	This implies the inequality.
\end{proof}

\begin{definition}[Basic Cell]
\label{definition:bas}
	Suppose $C$ is a cluster of height $i$, and $A$ is a crossing component of $C$.
	Define $l := w(A)$.
Define the basic cells of $A$ in $C$, $\Bas_{A}(C)$, to be the collection of the $h(i, l)$-cells of $C$ that \emph{intersect} $A$.
Define the basic cells of $C$, $\Bas(C)$, to be the union of $\Bas_A(C)$ for all crossing components $A$ of $C$.
\end{definition}

\begin{definition}[Owner of a Basic Cell]
	For some cluster $C$, define the owner of $e \in \Bas(C)$ to be the minimum weight crossing component $A$ such that
	$e \in \Bas_{A}(C)$.
\end{definition}

\begin{definition}[Promoted Cell and Virtual Cell]
	\label{definition:pro_vir}
	Suppose $C$ is a cluster of height $i$. Let $S$ be the set of sub-clusters of $C$
	that is not in $\Bas(C)$ but has a sibling in $\Bas(C)$.

	Consider each $e\in S$.
	\begin{compactitem}
		\item If there exists a sub-cluster $C'$ of $C$ such that $e\in \Bas(C')$, then
	define $\Pro_{e}(C) := \Des(e) \cap \Bas(C')$, and define $\Vir_{e}(C) := \emptyset$,
	where $C' \subset C$ is any one that satisfies $e \in \Bas(C')$.
		\item Otherwise,
	define $\Pro_{e}(C) := \emptyset$, and define $\Vir_{e}(C) := e$.
	\end{compactitem}

	Finally, $\Pro(C) := \bigcup_{e \in S}{\Pro_e(C)}$, and $\Vir(C) := \bigcup_{e \in S}{\Vir_{e}(C)}$,
	and elements in $\Pro(C)$ and $\Vir(C)$ are called promoted cells and virtual cells respectively.
\end{definition}

\begin{lemma}
\label{lemma:virtual_cells_safe}
For any cluster $C$, if $e \in \Vir(C)$,
then for any cluster $C' \subset C$ ($C'$ may equal $C$), $e\backslash\{ e' \in \Bas(C') \mid e' \subsetneq e \}$
has no intersection with any crossing component of $C'$.
\end{lemma}
\begin{proof}
Suppose not. Then, there exists a cluster $C' \subset C$, and a crossing component $A$ of $C'$,
such that $A$ intersects $u := e \backslash\{ e' \in \Bas(C') \mid e' \subsetneq e\}$.
This implies that there exists $u' \in \Bas_{A}(C')$, such that $e \subset u'$.
By Lemma~\ref{lemma:h_increase_one}, and the fact that $h(\Ht(C'), w(A)) \geq s^{\Ht(e)}$ and that
$h(0, w(A)) < s^{\Ht(e)}$, we know that there exists a cluster $C'' \subset C' \subset C$, such that $e \in \Bas_{A}(C'')$.
This contradicts with the definition of virtual cells.
\end{proof}

\begin{definition}[Non-basic Cell]
	\label{definition:nbas}
	We define the non-basic cells $\NBas(C)$ for a cluster $C$.
	If $C$ is the root cluster, then $\NBas(C) = \Pro(C) \cup \Vir(C) \backslash \Bas(C)$.
	For any other cluster $C$, 
	define $\NBas(C) := \{ e \cap C \mid e \in \Pro(C) \cup \Vir(C) \cup \NBas(\Par(C)) \backslash \Bas(C)\}$.
\end{definition}

\begin{definition}[Effective Cell]
	\label{definition:eff}
	For a cluster $C$, define the effective cells of $C$ as $\Eff(C) := \Bas(C) \cup \NBas(C)$.
\end{definition}

\begin{definition}[Refinement]
	Suppose $S_1$ and $S_2$ are collections of clusters. We say $S_1$ is a refinement of $S_2$,
	if for any $e \in S_2$, either $e \in S_1$, or all child clusters of $e$ are in $S_1$.
\end{definition}

\begin{lemma}
	\label{lemma:refinement}
	Suppose $C$ is a cluster that is not a leaf. Define $\{ C_i \}_i$ to be the collection of
	all the child clusters of $C$. Then $\bigcup_{i}{\Eff(C_i)}$ is a refinement of $\Eff(C)$.
\end{lemma}
\begin{proof}
Define $S := \bigcup_{i}{\Eff(C_i)}$.
It is sufficient to prove that for any $e \in \Eff(C)$, either $e \in S$,
or all child clusters of $e$ are in $S$.

If $e \in \NBas(C)$ and $e \neq C$, then $e \in S$ follows from Definition~\ref{definition:nbas}
and Definition~\ref{definition:eff}.
If $e \in \NBas(C)$ but $e = C$, then also by Definition~\ref{definition:nbas} and Definition~\ref{definition:eff},
$C\cap C_i = C_i \subset \Eff(C_i)$, and this
implies that all child clusters of $e$ are in $S$.

Otherwise, $e \in \Bas(C)$, then by Lemma~\ref{lemma:h_increase_one}, we know that either $e \in S$,
or there exists $e' \subset e$ such that $\Ht(e') = \Ht(e) - 1$ and $e' \in S$.
Then all siblings of $e'$ are in $S$, by the definition of promoted cells and virtual cells.
This implies that all child clusters of $e$ are in $S$.
\end{proof}

\begin{definition}[Candidate Center]
	Suppose $C$ is a cluster of height $i$. The set of candidate centers of $C$, denoted as $\Can(C)$,
	is the subset of $\bigcup_{j = \log_{s}{\gamma_{0}^2 s^i}}^{i}{N_j}$ that may become a center of $C$'s child cluster
	in the hierarchical decomposition.
\end{definition}

\begin{lemma}
	\label{lemma:bound_can}
	For any cluster $C$, the centers of clusters in $\Eff(C)$ are chosen from $\Can(C)$,
	and $|\Can(C)| \leq \kappa$, where $\kappa := O(\frac{1}{\gamma_0})^{O(k)}$.
\end{lemma}
\begin{proof}
	We first prove that centers of cluster in $\Eff(C)$ are chosen from $\Can(C)$.
	\begin{compactitem}
		\item For $e \in \Bas(C)$, by the definition of the basic cells,
			we have $\Ht(e) \geq \log_{s}{\gamma_0 s^i}$.
		\item For $e\in \Pro(C)$, we have that $e$ is a basic cell of some cluster $C'$,
			and hence $\Ht(e) \geq \log_{s}{\gamma_{0}^2 s^i}$. 
		\item For $e \in \Vir(C)$, since it is a sibling of a basic cell, so $\Ht(e) \geq \log_{s}{\gamma_0 s^i}$.
		\item For $e\in \NBas(C)$, there is a cluster $C''$ such that $C \subset C''$ and $e \in \Pro(C'') \cup \Vir(C'')$.
	\end{compactitem}
	Hence $\Ht(e)\geq \log_{s}{\gamma_{0}^2 s^i}$. Therefore, centers of clusters in $\Eff(C)$ are in $\Can(C)$.

We then bound $|\Can(C)|$. Suppose $i := \Ht(C)$. Observe that a center of height $j \leq i$ that may become
a center of a child cluster of $C$ are contained in a ball of diameter $O(s^i)$. Moreover, $N_j$ is an $s^j$ packing.
Hence, by packing property, $|\Can(C)| \leq O(\frac{1}{\gamma_{0}})^{O(k)}$.
\end{proof}

\begin{lemma}
	\label{lemma:bound_eff}
	Suppose $\Eff$ is defined in terms of a solution that is $(m, r)$-light. Then for each cluster $C$, $|\Eff(C)| \leq\rho$,
	where $\rho := O(\log_{s}{\frac{1}{\gamma_0}}) \cdot r^2 \cdot O(\frac{s}{\gamma_1})^{O(k)}$.
\end{lemma}
\begin{proof}
Suppose $C$ is of height $i$. We give upper bounds for $|\Bas(C)|$, $|\Pro(C)|$, $|\Vir(C)|$ and $|\NBas(C)|$ respectively.

\noindent\textbf{Bounding $|\Bas(C)|$.}
Fix a crossing component $A$ of $C$, and suppose $l := w(A)$. We upper bound $|\Bas_{A}(C)|$.
\begin{compactitem}
	\item If $i \in I_{1}(l)$, then $\Bas_{A}(C)$ is a subset of $\gamma_1 s^i$-cells of $C$.
		By packing property, $|\Bas_{A}(C)| \leq O(\frac{1}{\gamma_1})^k$.
	\item If $i \in I_{2}(l)$, then $\Bas_{A}(C)$ is a subset of $\gamma_1 \lfloor l \rfloor_s$-cells of $C$.
		Since all the $\gamma_1 \lfloor l \rfloor_s$-cells that intersect $A$ are inside a ball of diameter $O(l)$,
		by packing property, $|\Bas_{A}(C)| \leq O(\frac{s}{\gamma_1})^{O(k)}$.
	\item If $i \in I_{3}(l)$, then $\Bas_{A}(C)$ is a subset of $\gamma_0 s^i$-cells of $C$.
		Since all the $\gamma_0 s^i$-cells that intersect $A$ are inside a ball of diameter $O(\frac{\gamma_0}{\gamma_1} s^i)$,
		by packing property, $|\Bas_{A}(C)| \leq O(\frac{1}{\gamma_1})^k$.
\end{compactitem}
Since the solution is $r$-light, there are at most $r$ crossing components. Therefore,
\begin{align*}
|\Bas(C)| \leq r\cdot O(\frac{s}{\gamma_1})^{O(k)}.
\end{align*}

\noindent\textbf{Bounding $|\Pro(C)|$ and $|\Vir(C)|$.}
Recall that for $e \in \Bas(C)$, and for $e' \notin \Bas(C)$ that is a sibling of $e$,
we either include $e$ to $\Vir(C)$, or
include $\Des(e) \cap \Bas(C')$ to $\Pro(C)$, for some sub-cluster $C'$ of $C$.
In either cases, the number of added elements is at most $r\cdot O(\frac{s}{\gamma_1})^{O(k)}$,
and we charge this to $e$.

We observe that for each $e \in \Bas(C)$, it has at most $O(s)^k$ siblings, by packing property.
Therefore, each $e$ is charged at most $O(s)^k$ times. We conclude that 
\begin{align*}
|\Pro(C) \cup \Vir(C)| \leq O(s)^k \cdot r^2 \cdot O(\frac{s}{\gamma_1})^{O(k)}.
\end{align*}

\noindent\textbf{Bounding $|\NBas(C)|$.}
Suppose $P$ is the set consisting of $C$ and all its ancestor clusters.
By definition, $\NBas(C)$ is a subset of the inside $C$ clusters of $\bigcup_{p \in P}{(\Pro(p) \cup \Vir(p))}$.

We shall first prove that if $\Ht(p) - \Ht(C) > 2\log_{s}{\frac{1}{\gamma_0}}$,
then there is no element in $\Pro(p) \cup \Vir(p)$ that can appear in $\NBas(C)$, for any $p \in P$.
Suppose not. Then there exists some $p$ such that $\Ht(p) - \Ht(C) > 2\log_{s}{\frac{1}{\gamma_0}}$.
Let $j := \Ht(p)$.
We observe that all elements in $\Pro(p) \cup \Vir(p)$ have
height at least $\log_{s}{\gamma_{0}^{2} s^j} = j - 2\log_{s}{\frac{1}{\gamma_{0}}}$,
by Definition~\ref{definition:pro_vir} and Definition~\ref{definition:bas}.
However, if some element in $\Pro(p) \cup \Vir(p)$ appears in $C$,
then it has height less than $j - 2\log_{s}{\frac{1}{\gamma_0}}$,
by $\Ht(C)<\Ht(p) - 2\log_{s}{\frac{1}{\gamma_0}}$.
This is a contradiction.
Therefore,
\begin{align*}
	|\NBas(C)|
	\leq O(\log_{s}{\frac{1}{\gamma_0}}) \cdot r^2 \cdot O(\frac{s}{\gamma_1})^{O(k)}.
\end{align*}
Hence $|\Eff(C)| \leq |\Bas(C)| + |\NBas(C)| \leq O(\log_{s}{\frac{1}{\gamma_0}}) \cdot r^2 \cdot O(\frac{s}{\gamma_1})^{O(k)}$.
\end{proof}

\begin{definition}[Disjointification]
	For any collection of clusters $S$,
	define $\Dis(S) := \{e \backslash
	\bigcup_{e' \in S : e' \subsetneq e}{e'}
	\}_{e \in S}$.
	We say $e$ is induced by $u$ in $S$, if $u\in S$ and $e = u \backslash
	\bigcup_{e' \in S : e' \subsetneq u}{e'}
	$, and the height of $e$ is defined as the height of $u$.
\end{definition}

\begin{definition}[Cell Property]
\label{definition:cell_property}
	Suppose $F$ is an $\SFP$ solution, and suppose $f$ maps a cluster $C$ to a collection of sub-clusters of $C$.
	We say that
	$f$ satisfies the cell property in terms of $F$ if for all clusters $C$,
	for all $e\in \Dis(f(C))$, there is at most one
	crossing component of $C$ in $F$ that intersects~$e$.
\end{definition}

\begin{lemma}[Structural Property]
	\label{lemma:struct_property}
	Suppose an instance has a $q$-sparse optimal net-respecting solution $F$.
	Moreover, for each $i \in [L]$, for each $u\in N_i$, point $u$
	samples $O(k \log{n})$ independent random radii as in 
	Definition~\ref{defn:single_decomp}.
	Then, with constant probability, there exists
	a configuration from the sampled radii that defines a hierarchical decomposition, under which
	there exists
	an $(m, r)$-light solution $F'$ that includes all the points in $F$,
	and $\Eff$ defined in terms of $F'$ satisfies the cell property,
where
	\begin{compactitem}
		\item $\expct{w(F')} \leq (1+O(\epsilon)) \cdot w(F)$,
		\item $m:=O(\frac{s k L}{\epsilon})^k$ and $r:=O(1)^k \cdot q \log_{s}{\log{n}} + O(\frac{k}{\epsilon})^k +
	O(\frac{s}{\epsilon})^k$.
	\end{compactitem}
\end{lemma}
\begin{proof}
	We observe that the argument in~\cite[Lemma 3.1]{DBLP:conf/stoc/BartalGK12}
	readily gives an $(m, r)$-light solution $\widehat{F}$ with the desired $m$ and $r$,
	and also satisfies
	$\expct{w(\widehat{F})} \leq (1+\epsilon) \cdot w(F)$.

	We shall first show additional steps with additional cost at most $\epsilon w(F)$ in expectation,
	so that $\Bas$ defined in terms of the resultant solution satisfies the cell property.
	And then, we shall show that this
	\emph{implies} $\Eff$ defined in terms of the resultant solution also
	satisfies the cell property (hence no more additional cost caused).

\vspace{10pt}

	\noindent\textbf{Maintaining Cell Property: Basic Cells.}
	For $i := L, L-1, L-2, \ldots, 0$, for each height-$i$ cluster $C$,
	we examine $e \in \Dis(\Bas(C))$ in the non-decreasing order of its height.
	If there are at least two crossing components that intersect $e$,
	we add edges in $e$ to connect all crossing components that intersect $e$.
	We note that each added edge connects two components in $F$, and edges added are of length at most $\Diam(e)$.
	At the end of the procedure, we define the solution as $F'$.
	We observe that $\Bas$ defined in terms of $F'$ satisfies the cell property.

	Recall that each added edge connects two components.
	We charge the cost of the edge to one of the components that it connects to.
	Moreover, after a rearrangement (at the end of the procedure), we can make sure each edge is charged to
	one of the components it connects to and each component is charged at most once. 

	\noindent\textbf{Bounding The Cost.}
	We shall show that for a fixed component $A$, the expected cost it takes charge of is at most $\epsilon \cdot w(A)$.
	Define $l := w(A)$.
	The expected cost that $A$ takes is at most the following (up to contant)
	\begin{equation*}
		\sum_{i=1}^{L}{\Pr[\text{$A$ takes an edge in a cell of height $i$}] \cdot s^{i+1}}.
	\end{equation*}
	Define $p_i := \Pr[\text{$A$ takes an edge in a cell of height $i$}]$.
	Then,
	\begin{align*}
		\sum_{i=0}^{L}{p_i \cdot s^{i+1}}
		&\leq \sum_{i: s^i \leq 2\gamma_1 l}{s^{i+1}} + \sum_{i: s^i > 2\gamma_1 l}{p_i s^{i+1}} \\
		&\leq O(\gamma_1 s) l + \sum_{i : s^i > 2\gamma_1 l}{p_i s^{i+1}} \\
		&\leq O(\epsilon) l + \sum_{i : s^i > 2\gamma_1 l}{p_i s^{i+1}}.
	\end{align*}

	Fix an $i$ such that $s^i > 2\gamma_1 l$, and we shall upper bound $p_i$.
	Suppose in the event corresponding to $p_i$, $A$ takes charge of an edge inside a cell $e$ that is
	a basic cell of some height-$h$ cluster. Note that $h$ and $e$ are random and recall that the edge is inside a cell of height $i$.
	We shall give a lower bound of $h$.
	\begin{claim}
		$s^h \geq \frac{s^i}{2 \gamma_0}$.
	\end{claim}
	\begin{proof}
		Define the weight of the owner of $e$ to be $l'$.
		We first show that $h$ must be in $I_{3}(l')$.
		By the procedure of maintaining cell property,
	we know that $l' \leq l$. 

		If $h \in I_{1}(l')$, then $ \lfloor l' \rfloor_s \geq s^h $, and $s^i  \leq 2\gamma_1 s^h$ by $e$ is of height $i$ and the choice of radius in the single-scale decomposition.
		This implies that
		$s^{i} \leq 2\gamma_1 s^h \leq 2\gamma_1 l$, 
		which cannot happen since we assume $s^i > 2\gamma_1 l$.

		If $h \in I_{2}(l')$, then $s^i \leq 2\gamma_1 \lfloor l' \rfloor_s$.
		This implies that $s^i \leq 2\gamma_1 l$, which cannot happen as well.
		
		Therefore, $h \in I_{3}(l')$. This implies that $2\gamma_0 s^h \geq s^i$.
	\end{proof}
	Since the event that the edge is taken by $A$ automatically implies that $A$ is cut by a height-$h$ cluster, 
	and the probability that $A$ is cut at a height-$j$ cluster is at most
	$O(k)\cdot \frac{l}{s^j}$ for $j\in [L]$,
	we conclude that
	\begin{equation*}
	p_i \leq \sum_{j : s^{j} \geq \frac{s^i}{2\gamma_0 }}{\Pr[\text{$A$ is cut at height $j$}]}
	\leq O(k) \cdot \sum_{j : s^{j} \geq \frac{s^i}{2\gamma_0}}{\frac{l}{s^j}} \leq O(\gamma_0 k)\cdot \frac{l}{s^i}.
	\end{equation*}
	Hence $\sum_{i : s^i > 2\gamma_1 l}{p_i s^{i+1}} \leq O(\gamma_0 k s L) \cdot l \leq O(\epsilon) l$.

\vspace{10pt}

	\noindent\textbf{Maintaining Cell Property: Effective Cells.}
	Next we show that $\Bas$ defined in terms of $F'$ satisfies the cell property
	implies that $\Eff$ defined in terms of $F'$ also satisfies the cell property.

	Fix a cluster $C$ and fix $e \in \Dis(\Eff(C))$. We shall prove that 
	there is at most one crossing component of $C$ that intersects $e$ in $F'$.
	Suppose $e$ is induced by $u$ in $\Eff(C)$. 


	\begin{lemma}
	\label{lemma:exist_c_prime}
		If there is no cluster $\widehat{C}$ such that $C \subset \widehat{C}$ and $u \in \Vir(\widehat{C})$,
		then there exists cluster $C'$ such that $u \in \Bas(C')$,
		$\Ht(C') \leq \Ht(C)$ and $\Eff(C)$ is a refinement of $\Des(u) \cap \Bas(C')$.
	\end{lemma}
	\begin{proof}
		If $u \in \Bas(C)$, then we define $C' = C$, and the Lemma follows.

		If $u \in \NBas(C)$, then there exists $C''$ such that $C \subset C''$ and $u \in \Pro(C'')$.
		This is by the definition of non-basic cells,
		and by the assumption that
		there is not cluster $\widehat{C}$ such that $C \subset \widehat{C}$ and $u \in \Vir(\widehat{C})$.
		Then by the definition of the promoted cells, there exists cluster $C'$ such that $u \in \Bas(C')$, $\Ht(C') < \Ht(C'')$,
		and $\Des(u) \cap \Bas(C') \subset \Eff(C'')$. Since $u \in \NBas(C) \subset \Eff(C)$ and by Lemma~\ref{lemma:refinement},
		we know that $\Eff(C)$ is a refinement of $\Des(u) \cap \Bas(C')$. Hence, it remains to show $\Ht(C') \leq \Ht(C)$.

		Suppose for contradiction that $\Ht(C') > \Ht(C)$, so $\Ht(C) < \Ht(C') < \Ht(C'')$.
		By the definition of non-basic cells, we have that $\NBas(C') \cap \Bas(C') = \emptyset$.
		Since $u \in \Bas(C')$, we know that $u \notin \NBas(C')$. However, this implies that $u \notin \NBas(C)$,
		which contradicts with the assumption that $u \in \NBas(C)$.
\end{proof}

If there exists cluster $\widehat{C}$ such that $u \in \Vir(\widehat{C})$ and $C \subset \widehat{C}$,
then by Lemma~\ref{lemma:virtual_cells_safe}, there is no crossing component of $C$ in $F'$ that intersects $e$.

Otherwise, there is no cluster $\widehat{C}$ such that $u \in \Vir(\widehat{C})$ and $C \subset \widehat{C}$.
By Lemma~\ref{lemma:exist_c_prime}, there exists a cluster $C'$ such that $u \in \Bas(C')$, $\Ht(C') \leq \Ht(C)$
and $\Eff(C)$ is a refinement of $\Des(u) \cap \Bas(C')$.
We pick any one of such $C'$.
Define $e' \in \Dis(\Bas(C'))$ as the one induced by $u$ in $\Bas(C')$.
Since $\Bas$ defined in terms of $F'$ satisfies the cell property,
there is at most one crossing component of $C'$ that intersects $e'$.
	\begin{lemma}
	\label{lemma:e_subset_e_prime}
		$e \subset e'$.
	\end{lemma}
	\begin{proof}
		Recall that $e \in \Dis(\Eff(C))$ is induced by $u$ in $\Eff(C)$, and $e' \in \Dis(\Bas(C'))$ is induced by $u$ in $\Bas(C')$.
		Then we can write $e = u\backslash P$ and $e' = u \backslash P'$ such that
		$P \subset \Eff(C)$ and $P' \subset \Bas(C')$.
		Since $P' = \Des(u) \cap \Bas(C')$, and $\Eff(C)$ is a refinement of $\Des(u) \cap \Bas(C')$,
		we know that $P' \subset P$. This implies that $e \subset e'$.
	\end{proof}
	
	Since $\Ht(C) \geq \Ht(C')$, any crossing component of $C$ is also a crossing component of $C'$. Moreover, Lemma~\ref{lemma:e_subset_e_prime} implies that
	$e \subset e'$. Hence, if there are two crossing components $A_1, A_2$ of $C$ that intersect $e$, then $A_1$ and $A_2$ are also crossing
	components of $C'$ and both of them intersect $e'$.
	However, this cannot happen since $\Bas$ satisfies the cell property, and there is at most one crossing component in $C'$
	that intersects $e'$.
	Therefore, there is at most one crossing component of $C$ that intersects $e$.
\end{proof}

\subsection{Dynamic Program}
\label{section:dp_alg}
Recall that the input of \DP is an instance that has a $q$-sparse optimal
net-respecting solution, where $q \leq O(s)^{O(k)}\cdot q_0$, by Lemma~\ref{lemma:subinstance} and Corollary~\ref{cor:threshold}.
In the \DP algorithm, 
$O(k\log{n})$ random radii are independently sampled for each $u \in N_i$, $i \in [L]$, and then
a dynamic programming based algorithm is used to find a near optimal \SFP
solution over all hierarchical decompositions defined by the radii.
In this section, we shall describe in detail the dynamic program and an algorithm that solves the dynamic program efficiently. For completeness, we shall also analyze the correctness of the dynamic program.

We first describe the information needed
to identify each cluster at each height.

\noindent \textbf{Information to Identify a Cluster.}
Each cluster is identified by the following information.
\begin{compactitem}
\item[1.] Height $i$ and cluster center $u \in N_i$.
This has $L \cdot O(n^k)$ combinations, recalling that
$|N_i| \leq O(n^k)$.
\item[2.] For each $j \geq i$, and $v \in N_j$ such
that $d(u,v) \leq O(s^j)$, the random radius chosen by $(v,j)$.
Observe that the space around $B(u,O(s^i))$
can be cut by net-points in the same or higher heights that are nearby
with respect to their distance scales.
As argued in~\cite{DBLP:conf/stoc/BartalGK12},
the number of configurations
that are relevant to $(u,i)$ is at most $O(k \log n)^{L \cdot O(1)^k} = n^{O(1)^k}$, where $L = O(\log_s n)$
and $s = (\log n)^{\Theta(\frac{1}{k})}$.
\item[3.] For each $j > i$,
which cluster at height $j$ (specified
by the cluster center $v_j \in N_j$)
contains
the current cluster at height $i$.
This has $O(1)^{kL} = n^{O(\frac{k^2}{\log \log n})}$ combinations.
\end{compactitem}

To define the dynamic program, we start by defining the entries. 

\noindent\textbf{Entries of \DP.}
We define entries as $(C, (R, Y), (\BAS, \NBAS), (g, P))$.
Define $U := \Dis(\BAS\cup \NBAS)$. We define the following
\emph{internal constraints} for entries, where the parameters
$m, r$ are as defined in Lemma~\ref{lemma:struct_property},
and $\rho$ is as defined in Lemma~\ref{lemma:bound_eff}.
\begin{compactitem}
	\item $C$ is a cluster. 
	\item $R$ is a subset of the $m$ pre-defined portals, such that $|R| \leq r$. This intends to denote the active portals.
	\item $Y \subset 2^{R}$ is a partition of $R$. We intend to use it to record the subsets of portals that are connected inside $C$.
	\item $\BAS$ and $\NBAS$ are collections of sub-clusters of $C$ such that $\BAS \cap \NBAS = \emptyset$ and
		$|\BAS \cup \NBAS| \leq \rho$, and the centers of the clusters in $\BAS \cup \NBAS$ are chosen from $\Can(C)$.
		Moreover, $e \in \BAS$ implies that any sibling cluster of $e$ is in $\BAS \cup \NBAS$.
		We intend to use this to record the basic cells and non-basic cells.
	\item $g$ is a mapping from $U$ to $2^{Y}$. For some $e \in U$, we intend to use $g(e)$ to denote the portals that $e$ connects to inside $C$.
	\item $P \subset 2^{Y}$ is a partition of $Y$, such that $\forall e \in U$, $g(e) = Q$ implies that $Q$ is a subset of a part in $P$.
		The intended use of $P$ is to denote the portals that are to be connected outside $C$.
\end{compactitem}
We only consider the entries that satisfy the internal constraints.
We capture the intended use of an entry formally as follows.

\begin{definition}[Compatibility]
	Suppose $F$ is a graph on the metric space, and $E$ is an entry.
	Let $E := (C, (R, Y), (\BAS, \NBAS), (g, P))$.
	Define $F' := F|_{C \cup R}$. We say $F$ is \emph{compatible} to $E$, if $F'$ satisfies the following.
	\begin{compactitem}
	\item[1.] A part $y$ is in $Y$, if and only if $F'$ connects all the portals in $y$.
	\item[2.] $\BAS$ covers all components of $F'$ that intersect $R$.
	\item[3.] For $e \in U$, $g(e)$ is exactly the collection of subsets of $Y$ that $e$ is connected to by $F'$.
	\item[4.] Every terminal in $C$ is visited by $F'$.
	\item[5.] Every isolated terminal of $C$ is connected to at least one portal in $R$ by $F'$.
	\item[6.] Every terminal pair that both lie in $C$ is either in the same component of $F'$,
		or they are connected to $y_1$ and $y_2$ in $Y$ by $F'$ and $\{ y_1, y_2 \}$ is a subset of a part in $P$.
	\end{compactitem}
\end{definition}

We bound the number of entries in the following lemma.

\begin{lemma}[Number of Entries]
\label{lemma:num_entries}
	There are at most $O(n^{O(1)^k})\cdot O(\kappa mr)^{O(k)^k \cdot \rho r}$ number of entries.
	Moreover, for any fixed cluster $C$, the number of entries with
	$C$ as the cluster is at most $O(\kappa mr)^{O(k)^k \cdot \rho r}$. ($\kappa$ is defined as in Lemma~\ref{lemma:bound_can}.)
\end{lemma}
\begin{proof}
	Since $R$ is a set of at most $r$ portals chosen from $m$ pre-defined portals,
	there are at most $O(m^r)$ possibilities of $R$.
	Then after $R$ is fixed, there are $O(r^r)$ possibilities of $Y$, since $Y$ is a partition of $Y$ and $|R| \leq r$.

	To count the number of $\BAS$ and $\NBAS$, we count the union
	$S := \BAS\cup \NBAS$ of them, and then for any fixed $S$ we count the number of ways to assign elements
	in $S$ to $\BAS$ and $\NBAS$.
	Since it is required that the centers of clusters in $S$ are chosen from $\Can(C)$,
	to form $S$, we first choose at most $\rho$ centers from $\Can(C)$.
	There are at most $O(\kappa^{\rho})$ possibilities for this, by Lemma~\ref{lemma:bound_can}.
	For each chosen center $u$ that is of height $i_u$,
	we count the number of configurations of the cluster $C_u$ centered at $u$.
	Since $C$ is already fixed, we only need to consider relevant radii for clusters of height
	less than $\Ht(C)$ and at least $i_u$. Since $u \in \Can(C)$,
	and for $j \geq i_u$ there are $O(1)^k$ clusters of height-$j$ can affect $u$,
	we conclude that there are at most $O(k \log{n})^{O(1)^k \cdot \log_{s}{(\frac{1}{\gamma_{0}^2})}} \leq O(k \log{n})^{O(k)^k}$
	configurations for $C_u$.
	Since $|S|\leq \rho$, there are at most $O(k \log{n})^{O(k)^k \cdot \rho}$ configurations for all clusters
	in $S$, for any given the centers.
	Therefore, there are $O(\kappa)^{O(k)^k \cdot \rho}$ possibilities for $S$ in total.
	Then we assign elements in $S$ to one of $\BAS, \NBAS$, and there are at most $2^{|S|} \leq 2^{\rho}$ number of them.
	In conclusion, the number of possibilities for $\BAS$ and $\NBAS$ is at most
	$O(\kappa)^{O(k)^k\cdot \rho}$.

	With $S$ fixed, we count the number of possibilities of $g$. Since $g$ is a mapping from $U$ to $2^Y$,
	the number of such a mapping is at most $O((2^{|Y|})^{|U|}) \leq O(2^{\rho \cdot r})$.
	Finally, observe that $P$ is a partition of $Y$, and $|Y| \leq r$.
	This implies that $P$ has at most $O(r^r)$ possibilities.

	Therefore, after fixing $C$, there are at most $O(\kappa mr)^{O(k)^k \cdot \rho r}$ possibilities for $((R, Y), (\BAS, \NBAS), (g, P))$.

	We then count the number of possibilities of $C$. Observe that there are $O(n)^{O(k)}$ centers for $C$.
	For a fixed center, since the number of configurations is at most $n^{O(1)^k}$,
	we conclude that there are at most $O(n^{O(1)^k})\cdot O(\kappa mr)^{O(k)^k \cdot \rho r}$ entries in total.
\end{proof}

After we define the entries, we shall (recursively) define the value that is associated with each entry.
The intended value of an entry $E$ is the weight of the minimum graph that is \emph{recursively compatible}
to $E$ (see definition~\ref{definition:recursively_compatible}).

\begin{definition}[Child Entry Collection]
Suppose $E := (C, (R, Y), (\BAS, \NBAS), (g, P))$ is an entry.
We say a collection of entries $\{ (C_i, (R_i, Y_i), (\BAS_i, \NBAS_i), (g_i, P_i)) \}_i$
is a child entry collection of $E$, if $\{C_i\}_i$ is a partition of $C$ with $\Ht(C_i) = \Ht(C) - 1 $ for all $i$.
\end{definition}

\begin{definition}[Portal Graph]
We say a graph $G$ is a portal graph of a collection of entries $I := \{ (C_i, (R_i, Y_i), (\BAS_i, \NBAS_i), (g_i, P_i))\}_i$,
if the vertex set of $G$ is $\bigcup_{i}{R_i}$.
\end{definition}

\begin{definition}[Consistency Checking]
	\label{definition:consistency_checking}
Suppose $E:=(C, (R, Y), (\BAS, \NBAS), (g, P))$ is an entry, and $I := \{(C_i, (R_i, Y_i), (\BAS_i, \NBAS_i), (g_i, P_i))\}_i$
is a child entry collection of $E$ and $G$ is a portal graph of $I$.
We say $G$ and $I$ are \emph{consistent} with $E$,
if all checks in the following procedure are passed.
\begin{compactitem}
	\item[1.] Check if $\bigcup_{i}{(C_i \cup R_i)} = C \cup R$.
	\item[2.] We shall define $Y'$ to be a partition of $R' := \bigcup_{i}{R_i}$. 
		Initialize $Y' := \bigcup_{i}{Y_i}$, and whenever
		there are $y_1, y_2 \in Y'$ connected by $G$ or $y_1 \cap y_2 \neq \emptyset$,
		replace them by the union of them. Check if $Y'$ restricted to $R$ is exactly $Y$.
	\item[3.] For each $e \in \BAS$, check if there exists $i$ and $e' \in \BAS_i$,
		such that $e' = e$ or $e'$ is a child cluster of $e$.
	\item[4.] For each $e \in \NBAS$, check if either there exists $i$ and $e' \in \BAS_i \cup \NBAS_i$
		such that $e = e'$, or all child clusters of $e$ are in $\bigcup_{i}{(\BAS_i \cup \NBAS_i)}$.
	\item[5.] Define $g_{i}'$ to be a mapping from $U_i$ to $2^{Y}$, where
		$g_{i}'(e) := \{ y \cap R \mid y \in Y' \land \exists y' : (y' \in g_{i}(e) \land  y \cap y' \neq \emptyset) \}$, for $e \in U_i$.
		Here $g_{i}'(e)$ intends to mean the parts in $Y$ that $e$ connects to,
		which is defined by ``extending'' $g_{i}(e)$ with respect to $G$.
		For each $i$ and $u \in \BAS_i$, if there exists $e\in U_i$ such that $e \subset u$ and $g_{i}'(e) \neq \emptyset$,
		then check if there exists $u' \in \BAS$ such that $u = u'$ or $u$ is a child cluster of $u'$. 
	\item[6.] Define a mapping $g'$ from $U$ to $2^{Y}$,
		where $g'(e) := \bigcup_{i}{\bigcup_{e' \in U_i: e'\subset e}{g_{i}'(e')}}$,
		for $e \in U$. Check if $g'$ is exactly $g$. We observe that here we consider $e' \subset e$ only, and we shall see later
		why this is sufficient.
	\item[7.] For each $i$, for each $y_1 , y_2 \in Y_i$ ($y_1 \neq y_2$) such that $y_1, y_2$ are in the same part of $P_i$,
		check if either there exists $y \in Y'$ such that $y_1 \cup y_2 \subset y$,
		or there exists $y_{1}', y_{2}' \in Y'$ such that $y_{1}' \neq y_{2}'$, $y_1 \subset y_{1}'$, $y_2 \subset y_{2}'$, $y_{1}' \cap R \neq \emptyset$,
		$y_{2}' \cap R \neq \emptyset$, and $\{y_{1}' \cap R, y_{2}' \cap R\}$ is a subset of a part in $P$.
		This intends to check if the parts in $P_i$ are connected by $G$, or the information in $P_i$'s is passed to $P$.
	\item[8.] For each terminal pair $(a, b)$ such that $a \in C_i$ and $b \in C_j$ for $i\neq j$,
		suppose $a \in e_i$ and $b \in e_j$ for $e_i \in U_i$ and $e_j \in U_j$.
		Check if $g_i(e_i)$ is connected by $G$ to $g_j(e_j)$,
		or if $g_{i}'(e_i) \neq \emptyset$, $g_{j}'(e_j)\neq \emptyset$,
		$ g_{i}'(e_i) \cup g_{j}'(e_j)  $ is a subset of a part in $P$. This intends to check if $(a, b)$ are already connected by $G$,
		or otherwise they will be connected outside $C$.
	\item[9.] For each isolated terminal $a$ in $C$, check if there exists $i$ and $e\in U_i$,
		such that $a \in e $ and $g_{i}'(e)$ is non-empty.
\end{compactitem}
\end{definition}

\begin{definition}[Recursive Compatibility]
	\label{definition:recursively_compatible}
	Suppose $E := (C, (R, Y), (\BAS, \NBAS), (g, P))$ is an entry, and $F$ is some graph on the metric space.
	$F$ is \emph{recursively compatible} with $E$, if there exists a set $S$ of entries
	with $E \in S$ and with a unique entry in $S$ that corresponds to each descendant cluster of $C$,
	such that the following requirements hold.
	\begin{compactitem}
		\item For each $E' := (C', (R', Y'), (\BAS', \NBAS'), (g', P'))$ in $S$,
			we require $F' := F|_{C' \cup R'}$ be compatible to $E'$.
		\item For each $E' := (C', (R', Y'), (\BAS', \NBAS'), (g', P'))$ in $S$,
			suppose the child entry collection that consisting of elements in $S$ is $I'$,
			and define $I' := \{(C_t, (R_t, Y_t), (\BAS_t, \NBAS_t), (g_t, P_t))\}_t$.
			Define $G' := F|_{\bigcup_{t}{R_t}}$. (Note that $G'$ is a portal graph of $I'$.)
			We require $I'$ and $G'$ be consistent with $E'$.
	\end{compactitem}
\end{definition}

\noindent\textbf{Value of Entries.}
For any entry $E := (C, (R, Y), (\BAS, \NBAS), (g, P))$, we shall define its value $\Val(E)$.
The height-$0$ clusters are corresponding to the base cases. 
In particular, for any $C := \{x\}$ that is a height-$0$ cluster,
we define entries with such $C$ and with
$\BAS := \{ C \}$, $\NBAS := \emptyset$, 
$R := C$, $Y := \{ R \}$, $g(C) := Y$, $P := \{ Y \}$ to be the base entries.
All base entries have value $0$. All other (non-base) entries with height-$0$ clusters have
value $\infty$. 

We then define $\Val(E)$ when $\Ht(C) \neq 0$.
Define $\mathcal{I}_{E}$ to be the set of tuples $(I, G)$,
such that $I$ is a child entry collection of $E$ and $G$ is a portal graph of $I$, and $I, G$ are consistent.
The value of $E$ is defined as $\Val(E) := \min_{(I, G) \in \mathcal{I}_{E}}\{w(G) + \Val(I)\}$, where $\Val(I) = \sum_{E' \in I}{\Val(E')}$.
As we shall see in Lemma~\ref{lemma:min_compatible_value}, for any entry $E$, if $\Val(E)\neq \infty$,
then there actually exists a graph that is recursively compatible to
$E$ with weight $\Val(E)$.

\begin{lemma}[Counting $\mathcal{I}_E$]
	\label{lemma:bound_I_E}
	For any entry $E$, the number of possibilities of $\mathcal{I}_E$ is at most
	$O(k\log{n})^{O(s)^k} \cdot O(\kappa m r)^{O(sk)^{O(k)} \cdot \rho r^2}$,
	where $\kappa$ is defined as in Lemma~\ref{lemma:bound_can}.
\end{lemma}
\begin{proof}
	Define $E := (C, (R, Y), (\BAS, \NBAS), (g, P))$.
	We first bound the number of possibilities of
	child entry collections $I := \{ (C_i, (R_i, Y_i), (\BAS_i, \NBAS_i), (g_i, P_i)) \}_i$ of $C$.
	To define $I$, we start by defining $\{ C_i \}_i$.
	By packing property, there are at most $O(s)^k$ centers for the child clusters of $C$.
	For each center $u$ of the child cluster, there are at most $O(k\log{n})$ possible radii.
	Hence, there are at most $O(k\log{n})^{O(s)^k}$ possibilities for $\{C_i\}_i$.

	By Lemma~\ref{lemma:num_entries}, there are at most $Z$ possibilities
	for $((R_i, Y_i), (\BAS_i, \NBAS_i), (g_i, P_i))$ for any fixed $C_i$, where $Z := O(\kappa mr)^{O(k)^k \cdot \rho r}$.
	Therefore, there are at most
	$O(k\log{n})^{O(s)^k} \cdot Z^{O(s)^k}$ possibilities of $I$.
	
	For a fixed $I$, the vertex set of the portal graph $G$ of $I$ is fixed, and there are at most $O(s)^k \cdot r$
	vertices in $G$. Then the number of possibilities of $G$ for a fixed $I$ is at most the number of edge sets, and it is at most
	$2^{O(s)^{O(k)} \cdot r^2}$ since there are at most $(O(s)^k \cdot r)^2$ edges.

	In conclusion, there are at most
	$O(k\log{n})^{O(s)^k} \cdot Z^{O(s)^k} \cdot 2^{O(s)^{O(k)} \cdot r^2}$ possibilities
	of $\mathcal{I}_E$, which is at most $O(k\log{n})^{O(s)^k} \cdot O(\kappa m r)^{O(sk)^{O(k)} \cdot \rho r^2}$.
\end{proof}

\noindent\textbf{Final Entry.}
The final entry is the entry with
$C$ being the root cluster, $R$, $\BAS$, $\NBAS$ to be $\emptyset$,
and $Y, g, P$ being uniquely defined from $R, \BAS, \NBAS = \emptyset$.
We use the value of the final entry as the output of \DP.

\noindent\textbf{Evaluating The Final Entry}
Although we only care about the value of the final entry, it may be necessary to evaluate the value of other entries.
We shall define a (recursive) algorithm in Definition~\ref{definition:alg_value}
that takes an entry and returns the value of the input.
To get the value of the final entry which is the output of \DP,
we invoke the algorithm with the final entry as the input.

We note that the counting argument in Lemma~\ref{lemma:num_entries} and Lemma~\ref{lemma:bound_I_E} 
can both be naturally implemented as algorithms, with additional $O(n^{O(k)})$ factors in the running time
compared with the corresponding counting bounds. We will make use of these implementations as subroutines in 
Definition~\ref{definition:alg_value}. Moreover, the natural implementation of the consistency checking procedure in
Definition~\ref{definition:consistency_checking} runs in time $O(n^{O(k)})$.
\begin{definition}[Algorithm for Evaluating Value of Entries]
\label{definition:alg_value}
We define a recursive procedure that evaluates the value of an input entry $E := (C, (R, Y), (\BAS, \NBAS), (g, P))$.
	\begin{compactitem}
		\item If $\Ht(C) = 0$, then the value of it is already defined, and we return its value.
		\item If $\Ht(C)>0$ and $\Val(E)$ is already calculated, then we return the calculated value.
		\item Otherwise, $\Ht(C)>0$ and $\Val(E)$ has not yet calculated. The following procedure is executed.
			\begin{compactitem}
				\item[1.] Set the default value for $\Val(E) := \infty$.
				\item[2.] Calculate $\mathcal{I}_E$.
				\item[3.] For each element $(I, G) \in \mathcal{I}_E$, use
				the consistency checking procedure defined in Definition~\ref{definition:consistency_checking}
				to check if $I$ and $G$ are consistent with $E$.
				If they are consistent, then recursively use this procedure to calculate $\Val(I) + w(G)$,
				and update $\Val(E)$ if $\Val(E) + w(G)$ is smaller than $\Val(E)$.
				\item[4.] Finally, return $\Val(E)$ as the output.
			\end{compactitem}
	\end{compactitem}
\end{definition}

\begin{lemma}[Running Time]
	\label{lemma:running_time}
	The running time for the algorithm defined in Definition~\ref{definition:alg_value} is at most
	$O(n^{O(1)^k}) \cdot \exp(\sqrt{\log{n}}\cdot O(\frac{k}{\epsilon})^{O(k)})$.
\end{lemma}
\begin{proof}
	Suppose the input is $E := (C, (R, Y), (\BAS, \NBAS), (g, P))$.
	We observe that once the value for some entry is calculated,
	it would not be calculated again, and recalling the value takes constant time.
	Then we shall bound the time when $\Val(E)$ is not yet calculated and $\Ht(C)\neq 0$.

	Observe that for any given $I$ with $\Val(E')$ for all $E' \in I$ known and a graph $G$
	such that $(I, G) \in \mathcal{I}_E$,
	evaluating $\Val(I) + w(G)$ takes $O(n)^{O(k)}$ time.
	Therefore, combining with Lemma~\ref{lemma:num_entries} and Lemma~\ref{lemma:bound_I_E},
	there are at most $O(n^{O(1)^k}) \cdot Z$ entries,
	and it takes $O(n)^{O(k)} \cdot O(k\log{n})^{O(s)^k} \cdot O(\kappa m r)^{O(sk)^{O(k)} \cdot \rho r^2}$
	to evaluate each.
	In conclusion, the time for evaluating all the entries is at most
	$O(n^{O(1)^k}) \cdot O(k\log{n})^{O(s)^k} \cdot O(\kappa m r)^{O(sk)^{O(k)} \cdot \rho r^2}$.

	\noindent\textbf{Substituting Parameters.}
	Recall that we consider $q \leq O(s)^{O(k)}\cdot q_0$.
	Observe that
	$\frac{1}{\gamma_0} := \lceil \frac{1}{\hat{\gamma}_0} \rceil_s
	\leq O(\frac{k s^3 L}{\epsilon})$, and
	$\frac{1}{\gamma_1} := \lfloor \frac{1}{\hat{\gamma}_1} \rfloor_s
	\leq O(\frac{s^2}{\epsilon})$. Substituting $\gamma_0$ and $\gamma_1$,
	we have $\kappa \leq O(\frac{k s L}{\epsilon})^{O(k)}$ and
	$\rho \leq O(\frac{sk}{\epsilon})^{O(k)}$.
	Moreover,
	\begin{align*}
		r := O(1)^k \cdot q\log_{s}{\log{n}} + O(\frac{k}{\epsilon})^{k} + O(\frac{s}{\epsilon})^k 
		\leq O(\frac{sk}{\epsilon})^{O(k)},  
		m \leq O(\frac{sk L}{\epsilon})^k.
	\end{align*}
	By definition, $s := (\log{n})^{\frac{c}{k}}$, $L := O(\log_{s}{n}) = O(\frac{k \log{n}}{c \log{\log{n}}})$.
	Therefore, the running time is at most
	\begin{align*}
		&\quad  O(n^{O(1)^k}) \cdot O(k\log{n})^{O(s)^k} \cdot O(\kappa m r)^{O(sk)^{O(k)} \cdot \rho r^2} \\
		& \leq O(n^{O(1)^k}) \cdot O(k \log{n})^{O(s)^k} \cdot O(\frac{ksL}{\epsilon})^{O(\frac{sk}{\epsilon})^{O(k)}} \\
		& \leq O(n^{O(1)^k}) \cdot \exp(O(s)^{O(k)}\cdot O(\frac{k}{\epsilon})^{O(k)} \cdot \log{\frac{k \log{n}}{\epsilon}}) \\
		& \leq O(n^{O(1)^k}) \cdot \exp(O(\frac{k}{\epsilon})^{O(k)} \cdot O(\log{n})^{O(c)} \cdot \log{\log{n}}).
	\end{align*}
	By choosing constant $c$ to be sufficiently small so that $O(\log{n})^{O(c)} \cdot \log{\log{n}} \leq O(\sqrt{\log{n}})$,
	we conclude that the running time is at most $O(n^{O(1)^k}) \cdot \exp(\sqrt{\log{n}} \cdot O(\frac{k}{\epsilon})^{O(k)})$.
\end{proof}

\begin{lemma}[Characterizing the Value of Entries]
	\label{lemma:min_compatible_value}
For each entry $E :=  (C, (R, Y), (\BAS, \NBAS), (g, P))$ with $\Val(E) \neq \infty$,
$\Val(E)$ is the weight of the minimum weight graph
that is recursively compatible to the entry and uses points in $C\cup R$ only.
\end{lemma}
\begin{proof}
For the clusters of height $0$, the Lemma holds trivially.

Assuming the Lemma holds for all entries with the clusters of height $i-1$, 
we prove the Lemma for an entry $E$ with $C$ of height $i$ centered at $u\in N_i$,
where $i \geq 1$.
We shall first show that $\Val(E)$ is the weight of some graph that is recursively compatible to the entry
and uses points in $C\cup R$ only.
Then we shall show that the value is minimum.

\noindent\textbf{Feasibility.}
Suppose $(I, G) := \arg\min_{(I', G') \in \mathcal{I}_{E}}\{\Val(I') + w(G')\}$.
Define $I = \{ E_j \}_j$, where $E_j := (C_j, (R_j, Y_j), (\BAS_j, \NBAS_j), (g_j, P_j))$.
Since $\Val(E) \neq \infty$, we have $\Val(I') \neq \infty$. For $E_j \in I$, by assumption,
there exists a graph that is recursively compatible to $E_j$ and uses points in $C_j \cup R_j$ only, and we denote it as $F_j$.
We define a graph $F$ that is the union of $F_j$ for all $j$, and $G$.
Then $w(F) = \Val(I) + w(G)$.

We shall show that $F$ is recursively compatible to $E$. Since $(I, G) \in \mathcal{I}_{E}$, $I$ and $G$ are consistent with
$E$. Since $F_j$ is recursively compatible to $E_j$ for all $j$, it remains to verify $F$ is compatible to $E$.
When we say ``consistency checking procedure'', we refer to Definition~\ref{definition:consistency_checking}.
\begin{compactitem}
	\item $F$ uses points in $C \cap R$ only. This is by definition.
	\item A part $y$ is in $Y$, if and only if $F$ connects all the portals in the part $y$.
		This is by the step $2$ of the consistency checking procedure.
	\item $\BAS$ covers all components of $F$ that intersect $R$.
		This is by step $5$ of the consistency checking procedure.
	\item For $e \in U$, the collection of subsets of $Y$ that $e$ is connected to by $F$ is exactly $g(e)$.
		We note that
		step $3, 4$ of the consistency checking procedure,
		together with the internal constraint that $e' \in \BAS_j$ implies any sibling cluster of $e'$
		is in $\BAS_j \cup \NBAS_j$ for all $j$.
		This implies that 
		$\bigcup_{j}{(\BAS_j \cup \NBAS_j)}$ is a refinement of $\BAS \cup \NBAS$.
		Then, for each $j$, for each $e' \in U_j$, and for each $e \in U$,
		either $e' \subset e$ or $e' \cap e = \emptyset$. Therefore, step $6$ is sufficient to ensure this item.
		(If $e'$ is not a subset of $e$ but $e' \cap e \neq \emptyset$,
		then the $g_j$ mappings in the sub-entries have not sufficient
		information to determine the portals that $e' \cap e$ is connected to.)
	\item Every terminal in $C\cup R$ is visited by $F$.
	This is by the construction of $F$, and by $F_j$ is recursively compatible to $E_j$ for all $j$.
	\item Every isolated terminal of $C$ is connected to at least one portal in $R$ by $F$.
		This is by step $9$ of the consistency checking procedure.
	\item Every terminal pair that both lie in $C$ is either in the same component of $F$,
		or they are connected to $y_1$ and $y_2$ in $Y$ by $F$ and $\{ y_1, y_2 \}$ is a subset of a part in $P$.
		This is by step $7, 8$ of the consistency checking procedure, and by $F_j$ is recursively compatible to $E_j$ for all $j$.
\end{compactitem}
This implies that $F$ is recursively compatible to $E$.

\noindent\textbf{Optimality.}
Then we shall show that $\Val(E)$ is minimum.
Suppose not.
Define $l$ as the weight of the minimum weight graph
that is recursively compatible to $E$ and uses points in $C\cup R$ only.
Define $F'$ to be the corresponding graph recursively compatible to $E$ with weight $l$.
Since $F'$ is recursively compatible to $E$,
there exists $I' := \{ E_t := (C_t, (R_t, Y_t), (\BAS_t, \NBAS_t), (g_t, P_t))\}_t$
and a portal graph $G'$ of $I'$ that are consistent with $E$.
Moreover, there exists a graph $F_t$ that is recursively compatible to $E_t$, for all $t$.

We note that $(I', G') \subset \mathcal{I}_{E}$. 
Therefore, $\sum_{t}{w(F_t)} + w(G') = l < \Val(E) \leq \Val(I') + w(G')$.
This implies that $\sum_{t}{w(F_t)} < \sum_{E' \in I'}{\Val(E')}$, and hence
there exists $t$ such that $w(F_t) < \Val(E_t)$.
However, we know that $F_t$ is recursively compatible to $E_t$, and by assumption,
$\Val(E_t) \leq w(F_t)$.
This is a contradiction.
\end{proof}

\begin{corollary}
	\label{corollary:final_value}
	There exists a feasible solution to $\SFP$ whose weight is the value of the final entry.
\end{corollary}

\begin{lemma}[Good Solution is Recursively Compatible]
	\label{lemma:compatible_good_solution}
	Suppose for each $i \in [L]$ and $u \in N_i$, $O(k\log{n})$ radii are fixed.
	Suppose $F$ is an $(m, r)$-light solution such that $\Eff$ satisfies the cell property in terms of $F$
	under one of the hierarchical decompositions defined by the radii.
	Then the value of the final entry is at most $w(F)$.
\end{lemma}
\begin{proof}
	We shall show that $F$ is recursively compatible to the final entry,
	and then Lemma~\ref{lemma:min_compatible_value} implies that the value of the final entry is at most $w(F)$.

	Suppose we fix a hierarchical decomposition induced from the given radii,
	such that $F$ is $(m, r)$-light and $\Eff$ satisfies the cell property
	in terms of $F$.
	For each cluster $C$ in the decomposition, 
	we define $F_C := F_{C\cup R}$, 
	define an entry $E_C := (C, (R_C, Y_C), (\BAS_C, \NBAS_C), (g_C, P_C))$ as follows, where $R$ is the set of active portals for $C$.
	\begin{compactitem}
		\item $R_C := R$.
		\item $Y_C$ contains a part $y$, if and only if portals in $y$ is connected by $F_C$.
		\item $\BAS_C := \Bas(C)$, $\NBAS_C := \NBas(C)$.
		\item For each $e\in U_C$, let $g_C(e) := Q$, where $Q$ is the collection of parts in $Y_C$
		that $e$ is connected to by $F_C$.
		\item Define $P_C$ to be any one that satisfies
		\begin{compactitem}
			\item[1.] for each $e\in U_C$, $g_C(e) = Q$ implies $Q$ is a subset of $P_C$;
			\item[2.] for each terminal pair $(a, b)$ that both lie in $C$, if they are not connected by $F_C$ then the subsets of
				portals that $a$ and $b$ are connected by $F_C$ are in a same part of $P_C$.
		\end{compactitem}
	\end{compactitem}
	The internal constraints for an entry is satisfied,
	from the definition of the cells,
	the fact that $\Eff$ satisfies the cell property, Lemma~\ref{lemma:bound_can} and Lemma~\ref{lemma:bound_eff}.

	Then we (uniquely) define $I_C := \{ (C_i, (R_i, Y_i), (\BAS_i, \NBAS_i), (g_i, P_i)) \}_i$
	as the child collection of $E_C$, and define $G_C := F|_{\bigcup_{i}{R_i}}$ as the portal graph of $I_C$.
	
	We then check that $I_C$ and $G_C$ are consistent with $E_C$.
	Step $1, 2, 7, 8, 9$ are immediate. Step $3, 4, 5$ follow from the definition of the basic cells and non-basic cells.
	Inside step $6$, we observe that $g'(e)$ is evaluated by looking at $e' \subset e$ only (instead of considering all $e' \in U_i$), for $e' \in U_i$ for some $i$, and $e\in U$.
	However, this is indeed sufficient, since Lemma~\ref{lemma:refinement} asserts that for any $e \in U$, $e' \in U_i$ for any $i$,
	either $e' \subset e'$ or $e \cap e' = \emptyset$. 

	It remains to check the following for $E_C$, for each cluster $C$.
	\begin{compactitem}
		\item A part $y\in Y_C$, if and only if $F_C$ connects all the portals in the part $y$.
			This is by definition.
		\item $\BAS$ covers all components of $F_C$ that intersect $R_C$.
			This is by definition.
		\item For $e \in U_C$, the collection of subsets of $Y_C$ that $e$
			is connected to by $F_C$ is exactly $g_C(e)$.
			This is by definition.
		\item Every terminal in $C \cup R_C$ is visited by $F_C$. This is by the feasibility of $F$. 
		\item Every isolated terminal of $C$ is connected to at least one portal in $R_C$ by $F_C$. This is by the feasibility of $F$.
		\item Every terminal pair that both lie in $C$ is either in the same component of $F_C$,
		or they are connected to $y_1$ and $y_2$ in $Y_C$ by $F_C$ and $\{ y_1, y_2 \}$ is a subset of a part in $P_C$.
		This is by definition.
	\end{compactitem}
	This finishes the proof.
\end{proof}

Combining Lemma~\ref{lemma:struct_property}, Lemma~\ref{lemma:compatible_good_solution},
Corollary~\ref{corollary:final_value} and Lemma~\ref{lemma:running_time}, we conclude a PTAS for sparse \SFP instances.

\begin{corollary}[PTAS for Sparse \SFP Instances]
	\label{corollary:ptas_sparse_sfp}
	For an instance of $\SFP$ that has a $q$-sparse optimal net-respecting solution, algorithm \DP returns a $(1+\epsilon)$ solution with constant probability,
	running in time
	$O(n^{O(1)^k}) \cdot \exp(\sqrt{\log{n}}\cdot O(\frac{k}{\epsilon})^{O(k)})$, for $q \leq O(s)^{O(k)}\cdot q_0$.
\end{corollary}

\bibliography{main}


\end{document}